\PassOptionsToPackage{dvipsnames}{xcolor}
\documentclass[a4paper]{easychair}
\usepackage{doc}
\usepackage{mathtools} 
\usepackage{booktabs} 
\usepackage{tikz} 
\usepackage{mathtools}
\usetikzlibrary{backgrounds}
\usetikzlibrary{calc}
\usepackage[characters=true]{ltl}
\usepackage{todonotes}
\usepackage{amsthm}
\usepackage{xspace}
\usepackage{caption}
\usepackage{subcaption}

\newtheorem{counter}{Counter}
\newtheorem{definition}[counter]{Definition}
\newtheorem{proposition}[counter]{Proposition}
\newtheorem{corollary}[counter]{Corollary}
\newtheorem{theorem}[counter]{Theorem}
\newtheorem{lemma}[counter]{Lemma}
\newtheorem{example}[counter]{Example}
\DeclareMathOperator{\cf}{\LTLsquare\!\!\rightarrow}
\DeclareMathOperator{\cfmin}{\LTLsquare\!\!\rightarrow_{\!_\mathit{min}}}
\DeclareMathOperator{\mcf}{\LTLdiamond\!\!\rightarrow}
\DeclareMathOperator{\mcfmin}{\LTLdiamond\!\!\rightarrow_{\!_\mathit{min}}}
\DeclareMathOperator{\ucf}{\LTLhalfsquare\!\!\rightarrow}
\DeclareMathOperator{\ucfmin}{\LTLhalfsquare\!\!\rightarrow_{\!_\mathit{min}}}
\DeclareMathOperator{\emcf}{\LTLhalfdiamond\!\!\rightarrow}
\DeclareMathOperator{\emcfmin}{\LTLhalfdiamond\!\!\rightarrow_{\!_\mathit{min}}}

\makeatletter
\def\moverlay{\mathpalette\mov@rlay}
\def\mov@rlay#1#2{\leavevmode\vtop{%
		\baselineskip\z@skip \lineskiplimit-\maxdimen
		\ialign{\hfil$\m@th#1##$\hfil\cr#2\crcr}}}
\newcommand{\charfusion}[3][\mathord]{
	#1{\ifx#1\mathop\vphantom{#2}\fi
		\mathpalette\mov@rlay{#2\cr#3}
	}
	\ifx#1\mathop\expandafter\displaylimits\fi}
\makeatother

\newcommand{\cupdot}{\charfusion[\mathbin]{\cup}{\cdot}}

\newcommand{\hyperqptl}{\texttt{HyperQPTL}\xspace}
\newcommand{\qptl}{\texttt{QPTL}\xspace}
\newcommand{\qptlc}{\texttt{QPTL}$_\mathit{cf}$\xspace}
\newcommand{\ltl}{\texttt{LTL}\xspace}

\newcommand{\true}[0]{\mathit{true}}
\newcommand{\false}[0]{\mathit{false}}

\newcommand{\ldot}{\mathpunct{.}}

\newcommand{\U}{\LTLuntil}
\newcommand{\X}{\LTLnext}
\newcommand{\G}{\LTLglobally}
\newcommand{\F}{\LTLeventually}
\newcommand{\R}{\LTLrelease}

\newcommand{\Tr}{\mathit{T}}

\newcommand{\traceassign}{\Pi}
\newcommand{\tracevars}{\mathcal{V}}
\renewcommand{\models}{\vDash}
\newcommand{\nmodels}{\nvDash}

\newcommand{\donotshow}[1]{}

\usetikzlibrary{positioning,patterns}

\newcounter{mylabelcounter}
\makeatletter
\newcommand{\labelText}[2]{%
	#1\refstepcounter{mylabelcounter}%
	\immediate\write\@auxout{%
		\string\newlabel{#2}{{1}{\thepage}{{\unexpanded{#1}}}{mylabelcounter.\number\value{mylabelcounter}}{}}%
	}%
}
\makeatother

\title{Counterfactuals Modulo Temporal Logics\thanks{This work was partially supported by the DFG in project 389792660 (Center for Perspicuous Systems, TRR 248), and by the ERC Grant HYPER (No. 101055412).}}

\author{
Bernd Finkbeiner
\and
    Julian Siber
}

\institute{
  CISPA Helmholtz Center for Information Security,
  Saarbrücken, Germany\\
  \email{\{finkbeiner,julian.siber\}@cispa.de}
 }

\authorrunning{Finkbeiner and Siber}

\titlerunning{Counterfactuals Modulo Temporal Logics}

\begin{document}

\maketitle

\begin{abstract}
Lewis' theory of counterfactuals is the foundation of many contemporary notions of causality. In this paper, we extend this theory in the temporal direction to enable symbolic counterfactual reasoning on infinite sequences, such as counterexamples found by a model checker and trajectories produced by a reinforcement learning agent. In particular, our extension considers a more relaxed notion of similarity between worlds and proposes two additional counterfactual operators that close a semantic gap between the previous two in this more general setting. Further, we consider versions of counterfactuals that minimize the distance to the witnessing counterfactual worlds, a common requirement in causal analysis. To automate counterfactual reasoning in the temporal domain, we introduce a logic that combines temporal and counterfactual operators, and outline decision procedures for the satisfiability and trace-checking problems of this logic.
\end{abstract}

\section{Introduction}\label{sec:intro}

Evaluating counterfactual statements is a fundamental problem for many approaches to causal reasoning~\cite{MenziesB20}. Such reasoning can for instance be used to explain erroneous system behavior with a counterfactual statement such as `If the input $i$ at the first position of the observed computation $\pi$ had not been enabled then the system would not have reached an error $e$.' which can be formalized using the counterfactual operator $\cf$ and the temporal operator $\LTLeventually$: $$\pi \models (\lnot \mathit{i}) \cf (\lnot \LTLeventually \mathit{e}) \enspace .$$ Since the foundational work by Lewis~\cite{Lewis73a} on the formal semantics of counterfactual conditionals, many applications for counterfactuals~\cite{HalpernPearl05b,BeerBCOT12,KusnerLRS17,WachterMR17,BaierCFFJS21,CoenenDFFHHMS22} and some theoretical results on the decidability of the original theory~\cite{Lewis71} and related notions~\cite{EiterL02,AleksandrowiczC17} have been discovered. Still, certain domains have proven elusive for a long time, for instance, theories involving higher-order reasoning and an infinite number of variables. In this paper, we consider a domain that combines both of these aspects: temporal reasoning over infinite sequences. In particular, we consider counterfactual conditionals that relate two properties expressed in temporal logics, such as the temporal property $\lnot \LTLeventually \mathit{e}$ from the introductory example. Temporal logics are used ubiquitously as high-level specifications for verification~\cite{IEEE10,BaierK08} and synthesis~\cite{Finkbeiner16,MeyerSL18}, and recently have also found use in specifying reinforcement learning tasks~\cite{JothimuruganAB19,LiVB17}. Our work lifts the language of counterfactual reasoning to similar high-level expressions. We consider Quantified Propositional Temporal Logic (\texttt{QPTL}) because it can characterize the full class of $\omega$-regular properties and in this way subsumes popular specification languages. This results in our logic \qptlc, which mixes \qptl with counterfactual conditionals and can be used to check counterfactual dependencies between $\omega$-regular properties. We abstract away from any concrete causal models in this paper but refer to recent works by Halpern and Peters~\cite{HalpernP22}, and Coenen et al.~\cite{CoenenFFHMS22} on extending these models to infinitely many variables.

Focusing on the core counterfactual reasoning inherent to causality allows us to study several key problems arising in the temporal domain. We believe the main reason higher-order reasoning and an infinite number of variables induce difficulties for counterfactual reasoning is ultimately tied to Lewis' rejection of the \emph{Limit Assumption}, which stipulates that for any world $W$ and property $\varphi$, there is a unique set of worlds minimally close to $W$ that satisfy $\varphi$. If the assumption holds in some domain, counterfactual reasoning is simple: A statement such as `If $\varphi$ had held then $\psi$ would have held, too.', formally expressed by the formula $\varphi \cf \psi$, then would only need to compute the set of worlds minimally close to $W$ that satisfy $\varphi$ and check whether all of them satisfy $\psi$, too. However, as Lewis points out, the assumption is generally not true in any continuous domain, and as we will see later, it is also not generally true in the temporal domain. Previous works on defining notions of causality in settings with infinitely many variables sidestep the issue by restricting to similarity relations~\cite{HalpernPearl05a,CoenenFFHMS22} or logics for cause and effect~\cite{Leitner-FischerL13} that satisfy the Limit Assumption, but this imposes a significant toll on the precision of the inferred causes, since this requires coarsely overapproximating the set of closest traces.

One of the key insights of this work is that it is possible to reason about counterfactuals in the temporal domain even when rejecting the Limit Assumption. Without the assumption, evaluating counterfactual conditionals requires complex quantification over the possibly infinite chains of worlds ever closer to $W$. To solve these quantified statements in our domain, we use recent advances in the study of hyperproperties~\cite{ClarksonS10} and their corresponding temporal logics~\cite{ClarksonFKMRS14,Rabe16}, which originate in the verification of information-flow policies and allow to relate multiple traces of a system to another. This pushes their decidability to  the edge, for instance, the satisfiability problem of \hyperqptl, which extends \qptl with quantification over traces, is undecidable. While our counterfactual-temporal logic \qptlc has inherently relational semantics, its models are still traces, in contrast to the sets of traces modeled by \hyperqptl. Further, in \qptlc the trace-quantification is guarded by the counterfactual conditionals. We show that together this yields decidability for the satisfiability, model-checking and trace-checking problems of \qptlc by encoding them into the decidable model-checking problem of \hyperqptl .

We address several limitations  of Lewis' original theory of counterfactuals, with the goal that our logic \qptlc can be practically used to specify notions of temporal causality. For instance, the original theory requires a similarity relation between worlds that is total, i.e., for any two worlds $W_1$ and $W_2$ it needs to be possible to assess which one is closer to the reference world $W$. This turns out to be far too restrictive for reasonable similarity relations between infinite traces, and hence we extend Lewis' theory to non-total similarity relations. However, since in such relations, there may be several incomparable sets of worlds that may count as the worlds minimally close to $W$, this opens a semantic gap between Lewis' two proposed counterfactual operators. The crux is that, for instance, a naive extension of the `Would' counterfactual $\cf$ quantifies existentially over these sets of closest worlds, meaning its enough if there is one path from $W$ to worlds satisfying $\varphi$ where the closest worlds then also satisfy $\psi$. But there may be other paths to $\varphi$ where the closest worlds do not satisfy $\psi$. A similar problem exists with Lewis' semantics of the `Might' counterfactual. We argue that for both operators, the naive extensions  to non-total similarity relations do not match the intended semantics of the counterfactual statements they are supposed to formalize, and propose fixed semantics for non-total similarity relations. Another common requirement for causes is a notion of minimality~\cite{HalpernP01,Halpern15,Harbecke21,CoenenFFHMS22,CoenenDFFHHMS22}, such that negating the cause describes the minimal changes necessary to avoid the effect. This notion is not covered by Lewis' original theory. To enable our logic to express the minimality condition, we introduce \emph{minimal} counterfactual operators. The intuition behind a minimal counterfactual such as $\varphi \cfmin \psi$ is that it is meant to minimize the path from the reference world to the counterfactual worlds satisfying $\psi$. This boils down to a second-order requirement that quantifies over properties $\varphi'$ to see whether some of them characterize a superset of $\varphi$ and still qualify in the counterfactual. Since the second-order quantification is guarded by the minimal counterfactual operators, we can eliminate it by giving equisatisfiable first-order formulas that only quantify over traces.

We show that with an appropriate choice of underlying universe, \qptlc can express several notions of causality proposed in previous literature, and use it to forge an interesting link between event-based actual causation and property-based counterfactual causation.  

\paragraph{Contributions.} In summary, our contributions are as follows:
\begin{itemize}
	\item We extend Lewis' theory of counterfactual conditionals to non-total similarity relations (Section~\ref{sec:similarity}) by proposing two additional counterfactual conditionals that capture the intended semantics of `Would' and `Might' on these relations (Section~\ref{sec:nontotal}).
	\item We study a minimality criterion for counterfactuals that captures necessary reasoning for causal analysis and introduce minimal counterfactual operators (Section~\ref{sec:mincf}).
	\item We build a logic that mixes the classic counterfactuals due to Lewis and our newly proposed counterfactual conditionals with temporal properties expressed in \qptl, and show that the corresponding satisfiability and trace-checking problems are decidable (Section~\ref{sec:decidability}).
\end{itemize}
Necessary preliminaries on temporal logics are introduced in the following section (Section~\ref{sec:prelim}), and Lewis' original counterfactual conditionals are discussed in Section~\ref{sec:classic_cf}.

\paragraph{Related Work.}

Our theory is an extension of Lewis' theory of counterfactuals~\cite{Lewis73a,Lewis71} both in terms of the language of the antecedents and consequents, as well as for more general similarity relations and reasoning about minimality. Previous works in the context of axiomatizing causal modeling have extended the language of consequents to arbitrary Boolean formulas~\cite{Halpern00} as well as to counterfactual consequents~\cite{Briggs2012}. Our work lifts the language of cause and effect to restricted first- and second-order reasoning in an infinite domain and is in this way related to recent efforts by Halpern and Peters~\cite{HalpernP22} on causal reasoning with infinitely many variables. 

While to our knowledge there has been no previous work extending Lewis' logic of  counterfactuals to temporal reasoning, there are works that have proposed some notion of temporal causality in, e.g., Markov Decision Processes~\cite{ZiemekPFJB22} and reactive systems~\cite{CoenenFFHMS22}. Several previous works have made a connection between causality and hyperproperties~\cite{AbrahamB18,DimitrovaFT20,CoenenFFHMS22,CoenenDFFHHMS22}.

\section{Temporal Logics}\label{sec:prelim}

We consider temporal logics whose models are infinite \emph{traces} $t = t[0] t[1] \ldots \in (2^\mathit{AP})^\omega$ over some finite set of atomic propositions $\mathit{AP}$. As a basic temporal logic we consider Linear-time Temporal Logic (\texttt{LTL})~\cite{Pnueli77}. \ltl formulas are built with the following grammar, where $a \in \mathit{AP}$:
\begin{equation*}
\varphi \Coloneqq a \mid \neg \varphi \mid \varphi \land \varphi \mid \LTLnext \varphi \mid \varphi \U \varphi \enspace .
\end{equation*}
The semantics of \texttt{LTL} are given by the following satisfaction relation, which recurses over the positions $i$ of the trace.
\begin{equation*}
\begin{array}{lll}
t,i \models a       & \text{iff } & a \in t[i] \\
t,i  \models \neg \varphi              & \text{iff } & t,i  \nmodels \varphi \\
t,i  \models \varphi \land \psi         & \text{iff } & t,i  \models \varphi \text{ and } \pi_i \models \psi \\
t,i \models \X \varphi                & \text{iff } & t,i+1 \models \varphi \\
t,i \models \varphi\U\psi             & \text{iff } & \exists j \geq i \ldot t,j \models \psi~\land \forall i \leq k < j \ldot t,k \models \varphi
\end{array}
\end{equation*}
We say a trace $t$ satisfies a formula $\varphi$ iff the formula holds at the first position: $t,0 \models \varphi$, we also write $t \models \varphi$ to denote this. We also consider the usual derived Boolean ($\lor$, $\rightarrow$, $\leftrightarrow$) and temporal operators ($\varphi \R \psi \equiv \neg(\neg \varphi \U \neg \psi)$, $\F \varphi \equiv \true \U \varphi$, $\G \varphi \equiv \false \R \varphi$).

\begin{example}\label{ex:elevator}
    To illustrate how \texttt{LTL} can specify the dynamics of a system, consider an elevator that moves up ($u$) and down ($d$) between three floors bottom ($b$), middle ($m$), and top ($t$). We have a set of atomic propositions $\mathit{AP} = A \cup S$ composed of two subsets $A = \{u,d\}$ for actions and $S = \{b,m,t\}$ for states. The dynamics of the system starting at the lowest floor can be specified in an \ltl formula: 
    \begin{align*}
\varphi_{\mathit{elevator}} \equiv b \land \LTLglobally \big(&(b \land u \rightarrow \LTLnext m) \land (b \land d \rightarrow \LTLnext g) \land (m \land u \rightarrow \LTLnext t) \land (m \land d \rightarrow \LTLnext g)\,
\land\\&(t \land u \rightarrow \LTLnext t) \land (t \land d \rightarrow \LTLnext m)
\land (t \not\leftrightarrow b) \land (b \not\leftrightarrow m) \land (m \not\leftrightarrow t)\big) \enspace .
    \end{align*}
    The `Globally' operator $\LTLglobally$ universally quantifies over all time points in one sequence, requiring that all of the conjuncts in its body hold. The conjuncts themselves encode the dynamics, e.g., $b \land u \rightarrow \LTLnext m$ ensures that when the elevator moves up from the bottom floor, it reaches the middle floor in the next state (which the `Next' operator $\LTLnext$ accesses). The formulas in the last line encode that the elevator can only be on one floor at the same time, together with the others they also ensure that only one action can be done at any time point. The traces that satisfy the formula then describe all the valid dynamics of the system, e.g.:
    \begin{align*}
        t = \{b,u\}\{m,d\}(\{b,u\}\{m,d\})^\omega
    \end{align*}
    is a trace where the elevator cycles between the bottom and the middle floor. The $\omega$-superscript symbolizes that this part of the trace is repeated infinitely often.
\end{example}

\texttt{LTL} is of practical significance because its corresponding decision procedures are of comparatively low complexity. However, this comes at a cost in expressivity, such that it cannot specify that, e.g., an atomic proposition eventually holds at an odd position. To make our results applicable to as many properties as possible, we therefore consider Quantified Propositional Temporal Logic (\texttt{QPTL}), introduced by Sistla~\cite{Sistla83}, throughout the technical sections of this paper. \texttt{QPTL} extends \texttt{LTL} by quantification over atomic propositions. Its syntax is built atop of \texttt{LTL} as follows, where $q \not\in \mathit{AP}$ is a fresh atomic proposition, and $\varphi$ is an \texttt{LTL} formula:
\begin{equation*}
\psi \Coloneqq \exists q \ldot \psi \mid \forall q \ldot \psi \mid \varphi \enspace .
\end{equation*}
As presented by Finkbeiner et al.~\cite{FinkbeinerHHT20}, the semantics of the formulas quantifying over propositions can be stated using a replacement function $t[q \mapsto t_q]$ that given a trace $t \in (2^\mathit{AP})^\omega$ and a trace $t_q \in (2^{\{q\}})^\omega$ sets the occurrences of $q$ in $t$ to the ones in $t_q$, i.e., $t[q \mapsto t_q] =_{\{q\}} t_q$ and $t[q \mapsto t_q] =_{\mathit{AP} \setminus \{q\}} t$, where $t =_E t'$ means the truth value of the subset $E \subseteq \mathit{AP}$ agrees on all positions of the two traces.
\begin{equation*}
\begin{array}{lll}
t_i \models \exists q \ldot \psi      & \text{iff } & \exists t_q \in (2^{\{q\}})^\omega \ldot t[q \mapsto t_q] \models \psi \\
t_i  \models \forall q \ldot \psi             & \text{iff } & \forall t_q \in (2^{\{q\}})^\omega \ldot t[q \mapsto t_q] \models \psi \\
\end{array}
\end{equation*}
For a formula $\varphi$ in \qptl (and hence also in the fragment \ltl), we denote by $\mathcal{L}(\varphi)$ the set of traces that satisfy it.
\begin{example}
    Quantification over propositions allows limited forms of counting in \texttt{QPTL}, such as in the following formula $\varphi_{odd}$ that tracks the parity of positions with $q$ and hence can express that the atomic proposition $b$ eventually holds at an odd position:
    $$\varphi_{odd} \equiv \exists q \ldot \lnot q \land \LTLglobally (\lnot q \rightarrow \LTLnext q) \land \LTLeventually (q \land b) \enspace .$$
    Since on trace $\pi$ from Example~\ref{ex:elevator} $b$ only holds at even positions we have that $\pi \nmodels \varphi_{odd}$.
\end{example}

The semantics of the counterfactual conditionals we consider in this work require quantification over the worlds described by their antecedents and consequents. In our case, the worlds are traces. To express quantification over traces, we make use of hyperlogics~\cite{ClarksonFKMRS14,CoenenFHH19}, temporal logics that originated in information-flow security and relate multiple traces to one another. The hyper-counterpart to \qptl is \hyperqptl~\cite{Rabe16} and extends the syntax with trace quantifiers over a set of trace variables $\mathcal{V}$, where $\pi \in \mathcal{V}$. Since \hyperqptl is only used for decidability proofs, and all proofs are found in the appendix due to space reasons, we also state the syntax and semantics of \hyperqptl only in Appendix~\ref{app:hyperqptl}.

\section{Counterfactuals}\label{sec:cf}

We now outline our extended version of Lewis' theory of counterfactual conditionals. First, we extend Lewis' notion of a similarity relation to non-total orders (Section~\ref{sec:similarity}). Then, we introduce the classic counterfactuals (Section~\ref{sec:classic_cf}) and two new operators (Section~\ref{sec:nontotal}), and lastly, we consider a minimality criterion for counterfactual antecedents (Section~\ref{sec:mincf}). Since the concepts discussed in this section are not only applicable to the traces and temporal logics discussed in the previous section, we will adopt Lewis' modal nomenclature and speak of worlds and properties. We assume a set of worlds $\mathcal{U}$ called \emph{universe}. If not stated explicitly otherwise, all quantifiers in this section will quantify over $\mathcal{U}$. Further, we assume some logic $\mathcal{L}$ and a satisfaction relation that tells us for any world $W \in \mathcal{U}$, whether it satisfies some property $\varphi \in \mathcal{L}$, which we denote with $W \models \varphi$. We will make the connection to the previous section clear by using traces as worlds and \texttt{LTL} formulas as properties in our concrete examples.

\subsection{The Distance Between Worlds}\label{sec:similarity}

The semantics of counterfactuals rely on reasoning about the relative similarity of worlds with respect to the reference world in which the counterfactual is evaluated. Directly mapping Lewis' counterfactuals to our setting would necessitate a total preorder over the set of traces to tell us which of any two given traces is closer to our original trace. However, in practice, such a total order is unrealistic not just in our trace-based context, as changes between two worlds may simply be incomparable. For instance, consider that changing atomic proposition $u$ at position $0$ has the same quantitative distance as changing it at position $1$, but since the identity of these changes differs, the direction in the space of worlds is different. Many instances of counterfactual reasoning, therefore, base their notion of distance on subset relationships between changes, i.e., some world is further away than another if the changes necessary to obtain the former are a superset of the changes necessary to obtain the latter~\cite{HalpernP01,Halpern15,CoenenFFHMS22}. Then, if no subset relationship in either directions holds between the changes manifesting in two worlds, their distance is incomparable. The underlying spatial structure of such a similarity relation is a lattice over the equivalence classes of some preorder. We generalize Lewis' counterfactual reasoning to preorders to account for these more general similarity relations.
Formally, we require a \emph{comparative similarity relation} $\leq_W$, which is a preorder on $\mathcal{U}$ such that $W$ is a minimum: $\forall W' \ldot W' \not\leq_W W$.\footnote{Lewis refines his similarity relation based on a notion of accessibility. As our similarity relation is not necessarily total, accessibility can be easily encoded by not relating in $\leq_W$ the models inaccessible from $W$. We could then express accessibility of a world $W'$ from $W$ by requiring $W \leq_W W'$, but we will abstract away from this concept for simplicity and assume that $W$ is the unique minimum: $\forall W' \ldot W \leq_W W'$.} 

\begin{example}\label{ex:distance}
    In the context of trace logics as outlined in the previous section, worlds correspond to infinite traces such as $t$ in Example~\ref{ex:elevator}. Our universe may be given by the language of some \qptl formula such as the one describing the dynamics of the elevator system in Example~\ref{ex:elevator}, so we have  $\mathcal{U} = \mathcal{L}(\varphi_{\mathit{elevator}})$. A similarity relation may track the changes with respect to the reference trace $t$ over a subset $X$ of atomic propositions and can be formalized as follows:
    \begin{align*}
        \leq_t\!\!(\mathit{X}) = \{ (t_1, t_2) \mid \; &\forall n \in \mathbb{N} \ldot \forall x \in \mathit{X} \ldot  t[n] \neq_{\{x\}} t_1[n] \Rightarrow t[n] \neq_{\{x\}} t_2[n] \} \enspace .
    \end{align*}
    To illustrate with the elevator system from Example~\ref{ex:elevator}, the similarity relation $\leq_t\!\!(\mathit{A})$ orders a trace $t'' = \{b,d\}\{b,d\}(\{b,u\}\{m,d\})^\omega$, that changes the first two actions in trace $t$ but keeps the other actions the same, as closer to $t$ than trace $t' = (\{b,d\})^\omega$, which changes them on the whole sequence. We have $t'' \leq_t\!\!(\mathit{A}) \; t'$. Note that the above similarity relation is not total and would hence not be covered by Lewis' original theory.
    
\end{example}

\subsection{Classic Counterfactuals}\label{sec:classic_cf}

We start this section by recalling the semantics of Lewis' counterfactual operators $\cf$ and $\mcf$, based on the reformulation for similarity relations~\cite{Lewis73a}. In Lewis' theory, the operator $\cf$ is a formalization of `Would' counterfactual statements such as `If the elevator had eventually moved up two times in a row, then it would have reached the top floor.' The operator $\mcf$ is a formalization of `Might' counterfactual statements such as `If the elevator had eventually moved up two times in a row, then it might have reached the top floor two steps after the start.' Intuitively, both of the statements should be true on trace $t$ fro Example~\ref{ex:elevator}. No matter where we change the trace such that the elevator moves upwards twice, it will in all cases end up at the top floor. And there is one instance, i.e., when changing the first two actions appropriately, that it will be at the top floor at the third position. Formally speaking, the distinction between the two statements stems from different quantification over the closest worlds. `Would' counterfactuals are statements that quantify universally over all closest worlds that satisfy the antecedent, while `Might' counterfactuals quantify existentially. There are further subtle differences in the vacuous case which we will discuss after giving the formal semantics.
\begin{definition}[Semantics of `Would']\label{def:would_cf}
A world $W$ satisfies $\varphi \cf \psi$ iff:
$$ \forall W_1 \ldot W_1 \nmodels \varphi \; (\labelText{1}{lefteq:would})  \text{, or } \exists W_1 \ldot W_1 \models \varphi \land \forall W_2 \ldot  W_2 \leq_W W_1 \Rightarrow W_2 \models \varphi \rightarrow \psi \; (\labelText{2}{righteq:would}) \enspace .$$
\end{definition}
It is worth pointing out that the complex expression in Condition~\nameref{righteq:would} of Definition~\ref{def:would_cf} mainly stems from Lewis' rejection of the Limit Assumption, which poses that for any antecedent $\varphi$ and world $W$ there exists a unique set of (equally) closest worlds satisfying $\varphi$. Then it would be easy to simply quantify over this set universally and require the consequent to hold in all of the worlds. However, in many scenarios, this assumption, unfortunately, does not hold. Instead, there may in fact be an infinite chain of ever closer worlds that satisfy $\varphi$. In these instances, what we are rather interested in is finding a `threshold world' after which all closer worlds satisfying the antecedent $\varphi$ also satisfy the consequent $\psi$.

\begin{example}\label{ex:ftop}
    $\cf$ over \qptl formulas allows us to express the mix of counterfactual and temporal statements that is `If the elevator had eventually moved up two times in a row, then it would have reached the top floor.' It yields the following formula:
    $$\varphi_{top} \equiv \LTLeventually (u \land \LTLnext u) \cf \LTLeventually t \enspace .$$
    Interpreted in the universe given by $\mathcal{L}(\varphi_{\mathit{elevator}})$ and with the similarity relation $\leq_t\!\!(\mathit{A}) $, we have that trace $t = \{b,u\}\{m,d\}(\{b,u\}\{m,d\})^\omega$ from Example~\ref{ex:elevator} satisfies $\varphi_{top}$, because no matter where we change actions to ensure two moves upward in a row, we always end up in the top floor (this in fact holds for any trace in the universe). 
    
    For a more complex formula illustrating that the Limit Assumption does not hold in this setting, consider the statement `If the elevator would eventually only move downwards, then it would eventually stay on the bottom floor.' This corresponds to the formula:
    $$\varphi_{bottom} \equiv \LTLeventually (\LTLglobally d) \cf \LTLeventually (\LTLglobally b) \enspace .$$
    Trace $t' = (\{b,d\})^\omega$ satisfies $\LTLeventually (\LTLglobally d)$, but there is an infinite chain of traces closer to $t$ that also satisfy $\LTLeventually (\LTLglobally d)$:
    \begin{align*}
    &t'' = \{b,u\}\{m,d\}(\{b,d\})^\omega,~t''' = \{b,u\}\{m,d\}\{b,u\}\{m,d\}(\{b,d\})^\omega, \ldots
    \end{align*}
    and so on. Hence, we cannot avoid Lewis' complex quantification over traces to evaluate temporal counterfactuals in the general case.
\end{example}

As one can see from Condition~\nameref{lefteq:would} in Definition~\ref{def:would_cf}, a `Would' counterfactual is vacuously satisfied by a world if there are no related worlds that satisfy the antecedent. In contrast, a `Might' counterfactual strictly requires a world that satisfies the antecedent, mainly because Lewis bases the semantics of the two operators on the following duality law: $\varphi \cf \psi \equiv \lnot(\varphi \mcf \lnot \psi)$~\cite{Lewis73a}. This yields the following semantics for the `Might' counterfactual.

\begin{definition}[Semantics of `Might']\label{def:might_cf}
A world $W$ satisfies $\varphi \mcf \psi$ iff all of the following holds:
$$\exists W_1 \ldot W_1 \models \varphi~(\labelText{1}{lefteq:might})  \text{, and } \forall W_1 \ldot W_1 \models \varphi \Rightarrow \exists W_2 \ldot W_2 \leq_W W_1 \land W_1 \models \varphi \land \psi~(\labelText{2}{righteq:might}) \enspace .$$
\end{definition}

Again, significant complexity is introduced into the definition based on the rejection of the Limit Assumption. Here, however, the idea is not to find a `threshold world', but instead to find for any world in the chain of ever closer worlds another one that is closer (or equally close) such that both antecedent and consequent are true.

\subsection{Counterfactuals Over Non-total Similarity Relations}\label{sec:nontotal}

The semantics proposed by Lewis' work well if the similarity relation is a total order. However, as we can see in the following example, the semantics do not match the intuitive meaning of the operators when the similarity relation is not total, as in our setting.

\begin{example}\label{ex:problempreorders}
    Consider the statement `If the elevator had eventually moved up two times in a row, then it would have reached the top floor two steps after the start.',  which corresponds to the following formula:
    $$\varphi_{top}'' \equiv \LTLeventually (u \land \LTLnext u) \cf \LTLnext (\LTLnext t) \enspace .$$
    Intuitively, this statement should not be satisfied by trace $t$ from the previous examples. After all, there is trace
    $$t'_1 = \{b,u\}\{m,d\}\{b,u\}\{m,u\}\{t,d\}\{m,d\}t[6]t[7]\ldots$$
    that is a closest trace satisfying $\LTLeventually (u \land \LTLnext u)$ and it does not satisfy $\LTLnext (\LTLnext t)$. Yet, we can simply instantiate the existential quantifier in Condition~\nameref{righteq:would} of Definition~\ref{def:would_cf} with
    $$t'_2 = \{b,u\}\{m,u\}\{t,d\}\{m,d\}t[4]t[5]\ldots \enspace .$$
    Since there are no traces between $t'_2$ and $t$, the condition is satisfied, and we obtain that $t \models \varphi_{top}'' $. Now, consider
    $$\varphi_{top}''' \equiv \LTLeventually (u \land \LTLnext u) \mcf \LTLnext (\LTLnext t) \enspace .$$
    Intuitively, this statement should be satisfied, as there are indeed closest worlds that satisfy the consequent, namely $t_2'$. However, since the semantics of the operator quantifies over all worlds satisfying the antecedent and requires a smaller one that satisfies the consequent for each, the existence of $t_1'$ means that $t \nmodels \varphi_{top}'''$.
\end{example}

Example~\ref{ex:problempreorders} shows that on preorders the semantics of $\cf$ is too weak and that of $\mcf$ too strong to account for the intended meaning of the operators. Intuitively, `Would' should express that in the set of closest possible worlds, all worlds satisfy the consequent. However, in preorders there may exist multiple incomparable classes of worlds on different chains along the similarity relation. The focal question is how to quantify over these classes, existentially such that one class is enough (like  in the naive extension) or universally such that in all classes all models have to satisfy the consequent.

\begin{figure}[t]
    \centering
    \begin{subfigure}[b]{0.5\linewidth}
    \centering
    \scalebox{0.8}{
    \begin{tikzpicture}[node distance=3.5em,draw,thick]
    \node[draw,circle,minimum size=2em,fill=white](c) at (0,0){$W$};
    \node[draw,circle,minimum size=2em,above of=c,fill=white](n){};
    \node[draw,circle,minimum size=2em,below of=c,fill=white](s){};
    \node[draw,circle,minimum size=2em,left of=c,fill=white](w){};
    \node[draw,circle,minimum size=2em,right of=c,fill=white](e){};
    \node[draw,circle,minimum size=2em,below left of=c,fill=white](sw){};
    \node[draw,circle,minimum size=2em,below right of=c,fill=white](se){};
    \node[draw,circle,minimum size=2em,above left of=c,fill=white](nw){};
    \node[draw,circle,minimum size=2em,above right of=c,fill=teal](ne){};
    \node[draw,circle,minimum size=2em,above of=n,fill=white](nn){};
    \node[draw,circle,minimum size=2em,below of=s,fill=teal](ss){};
    \node[draw,circle,minimum size=2em,left of=w,fill=white](ww){};
    \node[draw,circle,minimum size=2em,right of=e,fill=white](ee){};
    \node[draw,circle,minimum size=2em,below left of=sw,fill=teal](sw2){};
    \node[draw,circle,minimum size=2em,below right of=se,fill=teal](se2){};
    \node[draw,circle,minimum size=2em,above left of=nw,fill=teal](nw2){};
    \node[draw,circle,minimum size=2em,above right of=ne,fill=teal](ne2){};
    \path[->,draw](c) -- (n);
    \path[->,draw](c) -- (s);
    \path[->,draw](c) -- (w);
    \path[->,draw](c) -- (e);
    \path[->,draw](c) -- (nw);
    \path[->,draw](c) -- (ne);
    \path[->,draw](c) -- (sw);
    \path[->,draw](c) -- (se);
    \path[->,draw](n) -- (nn);
    \path[->,draw](s) -- (ss);
    \path[->,draw](w) -- (ww);
    \path[->,draw](e) -- (ee);
    \path[->,draw](nw) -- (nw2);
    \path[->,draw](ne) -- (ne2);
    \path[->,draw](sw) -- (sw2);
    \path[->,draw](se) -- (se2);
    \path[<->,draw](ss) -- (se2);
    \path[<->,draw](ww) -- (nw2);
    \begin{scope}[on background layer]
    \draw [line width=4mm, Apricot, fill=Apricot] plot [smooth cycle] coordinates {(s.south) (sw.south west)(ww.south west)(nw2.north west)(nn.north)(ne.north east)(ee.north)(ee.east)(ee.south)};
    \draw [line width=1.5mm, Mahogany, fill=Mahogany] plot [smooth cycle,tension=0.5] coordinates {(c.north west)(c.north)(c.north east)(c.east)(c.south east)(sw.south east)(sw.south)(sw.south west)(sw.west)(sw.north west)};
    \end{scope}
    \end{tikzpicture}}
    \caption{`Would' and `Universal Would'.}\label{subfig:would}
    \end{subfigure}%
    \begin{subfigure}[b]{0.5\linewidth}
    \centering
    \scalebox{0.8}{
    \begin{tikzpicture}[node distance=3.5em,draw,thick]
    \node[draw,circle,minimum size=2em,fill=white](c) at (0,0){$W$};
    \node[draw,circle,minimum size=2em,above of=c,fill=white](n){};
    \node[draw,circle,minimum size=2em,below of=c,fill=white](s){};
    \node[draw,circle,minimum size=2em,left of=c,fill=white](w){};
    \node[draw,circle,minimum size=2em,right of=c,fill=white](e){};
    \node[draw,circle,minimum size=2em,below left of=c,fill=white](sw){};
    \node[draw,circle,minimum size=2em,below right of=c,fill=white](se){};
    \node[draw,circle,minimum size=2em,above left of=c,fill=white](nw){};
    \node[draw,circle,minimum size=2em,above right of=c,fill=teal](ne){};
    \node[draw,circle,minimum size=2em,above of=n,fill=white](nn){};
    \node[draw,circle,minimum size=2em,below of=s,fill=teal](ss){};
    \node[draw,circle,minimum size=2em,left of=w,fill=white](ww){};
    \node[draw,circle,minimum size=2em,right of=e,fill=white](ee){};
    \node[draw,circle,minimum size=2em,below left of=sw,fill=teal](sw2){};
    \node[draw,circle,minimum size=2em,below right of=se,fill=teal](se2){};
    \node[draw,circle,minimum size=2em,above left of=nw,fill=teal](nw2){};
    \node[draw,circle,minimum size=2em,above right of=ne,fill=teal](ne2){};
    \path[->,draw](c) -- (n);
    \path[->,draw](c) -- (s);
    \path[->,draw](c) -- (w);
    \path[->,draw](c) -- (e);
    \path[->,draw](c) -- (nw);
    \path[->,draw](c) -- (ne);
    \path[->,draw](c) -- (sw);
    \path[->,draw](c) -- (se);
    \path[->,draw](n) -- (nn);
    \path[->,draw](s) -- (ss);
    \path[->,draw](w) -- node(wm){} (ww);
    \path[->,draw](e) -- (ee);
    \path[->,draw](nw) -- node(nwm){} (nw2);
    \path[->,draw](ne) -- (ne2);
    \path[->,draw](sw) -- node(swm){} (sw2);
    \path[->,draw](se) -- (se2);
    \path[<->,draw](ss) -- (se2);
    \path[<->,draw](ww) -- (nw2);
    \begin{scope}[on background layer]
    \draw [line width=4mm, RedViolet, fill=RedViolet] plot [smooth cycle] coordinates {(s.south) (sw.south west)(w.south west)(nw.north west)(nn.north)(ne.north east)(ee.north)(ee.east)(ee.south)};
    \draw [line width=2mm, Lavender, fill=Lavender] plot [smooth cycle,tension=0.6] coordinates {(nw.north)(c.east)(sw.south)(swm.center)(wm.center)(nwm.center)};
    \end{scope}
    \end{tikzpicture}}
    \caption{`Might' and `Existential Might'.}\label{subfig:might}
    \end{subfigure}
    \caption{The set of worlds characterized by the counterfactuals on some similarity relation $\leq_W$. Nodes represent worlds (e.g., $W_{1,2,\ldots}$) where the center node is the reference world $W$. An edge from $W_1$ to $W_2$ means $W_1 \leq_W W_2$. For simplicity, we omit edges that can be inferred from reflexivity or transitivity of $\leq_W$. Nodes of worlds that satisfy $\psi$ are colored in teal.}
    \label{fig:cf_preorders}
\end{figure}
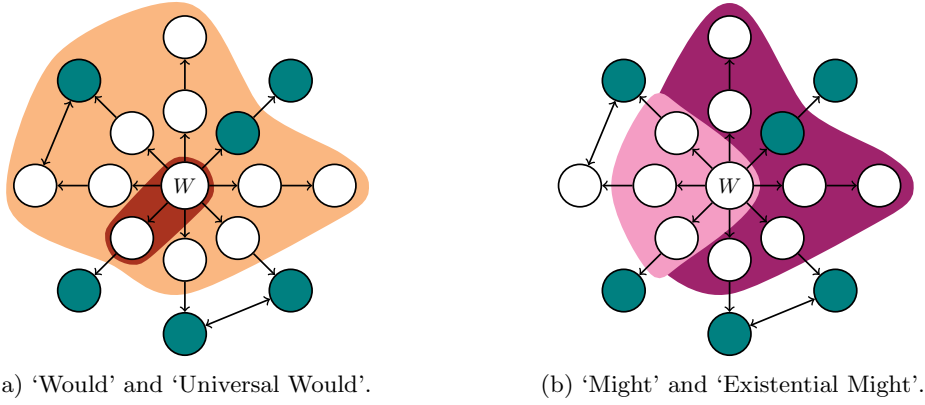

Figure~\ref{fig:cf_preorders} abstractly illustrates the problem with the `Would' counterfactual in Subfigure~\ref{subfig:would} and the problem with the `Might' counterfactual in Subfigure~\ref{subfig:might}. The areas colored in the dark colors represent the complement of the set characterized by an antecedent that satisfies the corresponding counterfactual in this universe and similarity relation, where the worlds colored in teal satisfy $\psi$. For example $W \models \varphi_{Would} \cf \psi$ if the worlds satisfying $\lnot \varphi_{Would}$ are the ones included in the dark brown area in Subfigure~\ref{subfig:would}. Here we can see that there are paths out of this area that lead to closest $\varphi_{Would}$-worlds that do not satisfy $\psi$, which clashes with the intended semantics of `If $\varphi_{Would}$ were true, then $\psi$ would be true as well.' What we instead want to capture is something like the bright brown area: No matter where we leave this area along the similarity relation, we always end up in a $\psi$-world. We capture this intention with the `Universal Would' counterfactual defined in the following, which requires all closest $\varphi$-worlds on any chain in the similarity relation to satisfy $\psi$.

\begin{definition}[Semantics of `Universal Would']\label{def:uwould_cf}
A world $W$ satisfies $\varphi \ucf \psi$ iff either of the following holds:
\begin{align*}
	&\forall W_1 \ldot W_1 \nmodels \varphi \text{, or }\\ &\forall W_1 \ldot W_1 \models \varphi \Rightarrow \exists W_2 \ldot W_2 \leq_W W_1 \land W_2 \models \varphi \land \forall W_3 \ldot W_3 \leq_W W_2 \Rightarrow W_3 \models \varphi \rightarrow \psi \enspace .
\end{align*}
\end{definition}

To ground our definition, we can show that it is equivalent to Lewis' classic `Would' counterfactual on total orders.

 \begin{proposition}\label{prop:equiv_ccf_mcf}
    If $\leq_W$ is a total order, then `Universal Would' and `Would' are equivalent, i.e., for all pairs of properties $\varphi$ and $\psi$ we have that $\varphi \ucf \psi \equiv \varphi \cf \psi$.
\end{proposition}

We now introduce a similar counterpart to the `Might' counterfactual. We derive the semantics from a duality law similar to the original one: $\varphi \ucf \psi \equiv \lnot (\varphi \emcf \lnot \psi)$. 

\begin{definition}[Semantics of `Existential Might']\label{def:emight_cf}
A world $W$ satisfies $\varphi \emcf \psi$ iff:
$$\exists W_1 \ldot W_1 \models \varphi \land \forall W_2 \ldot W_2 \leq_W W_1 \land W_2 \models \varphi \Rightarrow \exists W_3 \ldot W_3 \leq_W W_2 \land W_3 \models \varphi \land \psi \enspace .$$
\end{definition}

Since we derived the `Existential Might' counterfactual from a similar duality law as Lewis, we can deduce from Proposition~\ref{prop:equiv_ccf_mcf} that our `Existential Might' and the original `Might' are also equivalent for total similarity relations.

 \begin{corollary}\label{cor:equiv_emcf_mcf}
    If $\leq_W$ is a total order, then `Existential Might' and `Might' are equivalent, i.e., for all pairs of properties $\varphi$ and $\psi$ we have that $\varphi \emcf \psi \equiv \varphi \mcf \psi$.
\end{corollary}

Subfigure~\ref{subfig:might} illustrates how this definition captures the intended meaning of `Might' counterfactuals on non-total similarity relations. To illustrate the problem with Lewis' original semantics for `Might': The complement of some satisfying antecedent $\varphi_{Might}$ for a classic `Might' counterfactual is colored in dark violet. For satisfaction it is required that no matter where we leave this area, for any world there exists an equally close world satisfying $\psi$. This means for example, that we can include the worlds furthest to the west and north-west in $\varphi_{Might}$. This is quite strict, our `Existential Might' counterfactual illustrated in bright violet instead allows the closest worlds on a chain to not satisfy $\psi$, as long as on some chain there exists a closest world satisfying $\psi$.

\subsection{Minimal Counterfactuals}\label{sec:mincf}

When reasoning about causation, we are most often interested in some notion of minimality to characterize the minimal changes necessary to avoid a given effect~\cite{HalpernP01,Halpern15,Harbecke21,CoenenDFFHHMS22,CoenenFFHMS22}. From a counterfactual point of view, minimality formulates an additional condition on the antecedent $\varphi$ such that the property defines the largest set possible. The question of whether to term this notion minimality or maximality is a matter of perspective. On the one hand, the criterion maximizes the language of $\varphi$, but on the other hand, this in fact minimizes the amount of changes necessary to ensure $\varphi$, since more worlds in $\varphi$ mean more opportunities to move into $\varphi$ earlier when moving along the similarity relation. Since for causation this abstract criterion is usually called minimality, we adopt the same name here.

\begin{definition}[Minimal Counterfactuals]\label{def:mini}
Given a world $W$ and a counterfactual conditional $\rightsquigarrow \; \in \{ \cf,\mcf,\ucf,\emcf \} $, the minimal counterfactual conditional $\varphi \rightsquigarrow_{\!_\mathit{min}} \psi$ is true iff all of the following holds:
$$W \models \varphi \rightsquigarrow \psi ~(\labelText{1}{lefteq:mincf}) \text{, and }\nexists \; \varphi' \ldot (\varphi \rightarrow \varphi') \land (\varphi' \not\rightarrow \varphi)  \land W \models \varphi' \rightsquigarrow \psi~(\labelText{2}{righteq:mincf}) \enspace .$$
\end{definition}

\begin{example}
    Minimal counterfactuals ensure that the antecedent does not overspecify the changes necessary to get to the consequent. For instance, consider that both of the following statements hold:
    \begin{align*}
        t \models u \land \LTLnext u \ucf \LTLnext (\LTLnext t), \text{ and } t \models \LTLnext u \ucf \LTLnext (\LTLnext t) \enspace .
    \end{align*}
    While the antecedent in the lower statement is more concise because $t$ already has $u$ at the first position, neither is a minimal antecedent, because minimality additionally ensures that all possible antecedents are included, we have:
    $$ t \models (\LTLnext u) \lor \LTLnext (\LTLnext t) \ucfmin \LTLnext (\LTLnext t) \enspace .$$
    For a short argument, consider why there cannot be a trace in the minimal antecedent that does not satisfy $(\LTLnext u) \lor \LTLnext (\LTLnext t)$. Such a trace would have no upwards movement at the second position, and all traces closer to $t$ do neither. Since all traces start at the bottom floor, we know that none of the traces in between satisfy $\LTLnext (\LTLnext t)$. 
\end{example}

While it may seem odd that the consequent can be a necessary part of the antecedent like in the above example, we note that it is common for the effect to be in the set of its counterfactual causes~\cite{HalpernP01,Halpern15,Bochman18}. We believe there is a quite direct connection between minimal counterfactuals, and conjunctive and disjunctive causes from the causal modeling literature. We expand on this later in Section~\ref{sec:ex}.

As one can see in Definition~\ref{def:mini}, minimality is a second-order statement that quantifies over properties. Hence, the minimal counterfactual conditionals $\cfmin$, $\mcfmin$, $\ucfmin$, and $\emcfmin$ are essentially guarded second-order quantifiers. However, one of our main insights is that this second-order quantification inherent to these operators can be eliminated such that for any minimal counterfactual conditional there exists an equisatisfiable formula without quantification over properties. We prove this in the following lemma and will use this in the following section for concrete decidability results. 

\begin{lemma}\label{lem:qe_cf}
    For a minimal counterfactual statement $\varphi \rightsquigarrow_{\!_\mathit{min}} \! \psi$, with $\rightsquigarrow \; \in \{ \cf,\mcf,\ucf,\emcf \} $, there exists a parameterized formula $\varphi_{\mathit{FO}}(W)$, with $W \in \mathcal{U}$, that quantifies only over worlds and not properties such that for all worlds $W' \in \mathcal{U}$:
    $W' \models \varphi \rightsquigarrow_{\!_\mathit{min}} \psi$ iff $\varphi_{\mathit{FO}}(W')$ is valid.
\end{lemma}

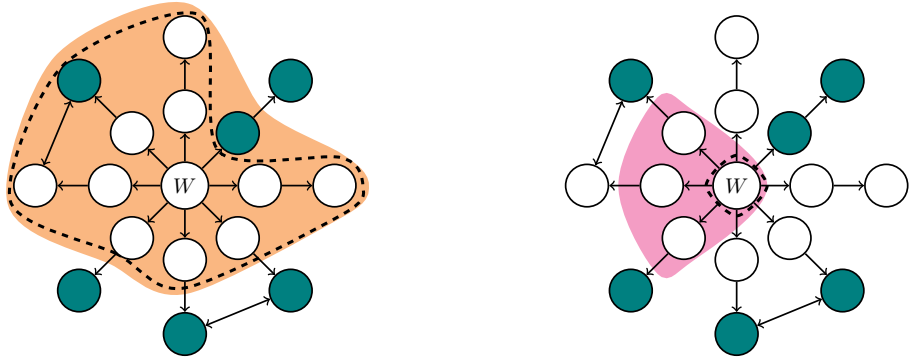
\begin{figure}[t]
    \centering
    \begin{subfigure}[b]{0.5\linewidth}
    \centering
    \scalebox{0.8}{
    \begin{tikzpicture}[node distance=3.5em,draw,thick]
    \node[draw,circle,minimum size=2em,fill=white](c) at (0,0){$W$};
    \node[draw,circle,minimum size=2em,above of=c,fill=white](n){};
    \node[draw,circle,minimum size=2em,below of=c,fill=white](s){};
    \node[draw,circle,minimum size=2em,left of=c,fill=white](w){};
    \node[draw,circle,minimum size=2em,right of=c,fill=white](e){};
    \node[draw,circle,minimum size=2em,below left of=c,fill=white](sw){};
    \node[draw,circle,minimum size=2em,below right of=c,fill=white](se){};
    \node[draw,circle,minimum size=2em,above left of=c,fill=white](nw){};
    \node[draw,circle,minimum size=2em,above right of=c,fill=teal](ne){};
    \node[draw,circle,minimum size=2em,above of=n,fill=white](nn){};
    \node[draw,circle,minimum size=2em,below of=s,fill=teal](ss){};
    \node[draw,circle,minimum size=2em,left of=w,fill=white](ww){};
    \node[draw,circle,minimum size=2em,right of=e,fill=white](ee){};
    \node[draw,circle,minimum size=2em,below left of=sw,fill=teal](sw2){};
    \node[draw,circle,minimum size=2em,below right of=se,fill=teal](se2){};
    \node[draw,circle,minimum size=2em,above left of=nw,fill=teal](nw2){};
    \node[draw,circle,minimum size=2em,above right of=ne,fill=teal](ne2){};
    \path[->,draw](c) -- (n);
    \path[->,draw](c) -- (s);
    \path[->,draw](c) -- (w);
    \path[->,draw](c) -- (e);
    \path[->,draw](c) -- (nw);
    \path[->,draw](c) -- (ne);
    \path[->,draw](c) -- (sw);
    \path[->,draw](c) -- (se);
    \path[->,draw](n) -- (nn);
    \path[->,draw](s) -- (ss);
    \path[->,draw](w) -- (ww);
    \path[->,draw](e) -- (ee);
    \path[->,draw](nw) -- (nw2);
    \path[->,draw](ne) -- (ne2);
    \path[->,draw](sw) -- (sw2);
    \path[->,draw](se) -- (se2);
    \path[<->,draw](ss) -- (se2);
    \path[<->,draw](ww) -- (nw2);
    \begin{scope}[on background layer]
    \draw [line width=4mm, Apricot, fill=Apricot] plot [smooth cycle] coordinates {(s.south) (sw.south west)(ww.south west)(nw2.north west)(nn.north)(ne.north east)(ee.north)(ee.east)(ee.south)};
    \draw [line width=0.5mm, draw,dashed] plot [smooth cycle] coordinates {($(s.south)  + (0,-0.1)$) (sw.south west)($(ww.south west)  + (-0.15,0)$)(nw2.north west)($(nn.north)  + (0.25,0)$)($(ne.south west)  + (-0.05,-0.05)$)($(ee.north)  + (0.1,0.1)$)($(ee.east)  + (0.1,0)$)($(ee.south)  + (0.15,-0.1)$)};
    \end{scope}
    \end{tikzpicture}}
    \caption{`Universal Would' and its minimal version.}\label{subfig:wouldmin}
    \end{subfigure}%
    \begin{subfigure}[b]{0.5\linewidth}
    \centering
    \scalebox{0.8}{
    \begin{tikzpicture}[node distance=3.5em,draw,thick]
    \node[draw,circle,minimum size=2em,fill=white](c) at (0,0){$W$};
    \node[draw,circle,minimum size=2em,above of=c,fill=white](n){};
    \node[draw,circle,minimum size=2em,below of=c,fill=white](s){};
    \node[draw,circle,minimum size=2em,left of=c,fill=white](w){};
    \node[draw,circle,minimum size=2em,right of=c,fill=white](e){};
    \node[draw,circle,minimum size=2em,below left of=c,fill=white](sw){};
    \node[draw,circle,minimum size=2em,below right of=c,fill=white](se){};
    \node[draw,circle,minimum size=2em,above left of=c,fill=white](nw){};
    \node[draw,circle,minimum size=2em,above right of=c,fill=teal](ne){};
    \node[draw,circle,minimum size=2em,above of=n,fill=white](nn){};
    \node[draw,circle,minimum size=2em,below of=s,fill=teal](ss){};
    \node[draw,circle,minimum size=2em,left of=w,fill=white](ww){};
    \node[draw,circle,minimum size=2em,right of=e,fill=white](ee){};
    \node[draw,circle,minimum size=2em,below left of=sw,fill=teal](sw2){};
    \node[draw,circle,minimum size=2em,below right of=se,fill=teal](se2){};
    \node[draw,circle,minimum size=2em,above left of=nw,fill=teal](nw2){};
    \node[draw,circle,minimum size=2em,above right of=ne,fill=teal](ne2){};
    \path[->,draw](c) -- (n);
    \path[->,draw](c) -- (s);
    \path[->,draw](c) -- (w);
    \path[->,draw](c) -- (e);
    \path[->,draw](c) -- (nw);
    \path[->,draw](c) -- (ne);
    \path[->,draw](c) -- (sw);
    \path[->,draw](c) -- (se);
    \path[->,draw](n) -- (nn);
    \path[->,draw](s) -- (ss);
    \path[->,draw](w) -- node(wm){} (ww);
    \path[->,draw](e) -- (ee);
    \path[->,draw](nw) -- node(nwm){} (nw2);
    \path[->,draw](ne) -- (ne2);
    \path[->,draw](sw) -- node(swm){} (sw2);
    \path[->,draw](se) -- (se2);
    \path[<->,draw](ss) -- (se2);
    \path[<->,draw](ww) -- (nw2);
    \begin{scope}[on background layer]
    \draw [line width=2mm, Lavender, fill=Lavender] plot [smooth cycle,tension=0.6] coordinates {(nw.north)(c.east)(sw.south)(swm.center)(wm.center)(nwm.center)};
    \draw [line width=0.5mm, draw, dashed] plot [smooth cycle,tension=0.5] coordinates {($(c.north) + (0,0.1)$)($(c.east)  + (0.1,0)$)($(c.south)  + (-0,-0.1)$)($(c.west)  + (-0.1,-0)$)};
    \end{scope}
    \end{tikzpicture}}
    \caption{`Existential Might' and its minimal version.}\label{subfig:mightmin}
    \end{subfigure}
    \caption{The set of worlds characterized by non-minimal counterfactuals (colored areas) compared to the set of worlds characterized by minimal counterfactuals (dashed line) in the same setting as in Figure~\ref{fig:cf_preorders}.}
    \label{fig:cf_min}
\end{figure}

To illustrate minimal counterfactuals consider the abstract scenario in Figure~\ref{fig:cf_min}. In Subfigure~\ref{subfig:wouldmin} colored areas correspond to the complements of a satisfying antecedent of a `Universal Would' counterfactual with reference world $W$. However, this antecedent is not minimal, because the complement of the dashed area is equally part of a `Universal Would' counterfactual satisfied by $W$. The illustration sums up the following intuition: The worlds encountered by leaving the dashed area that satisfy $\psi$ are now overall closer to $W$ than before. The same idea holds for the more minimal antecedent of the `Existential Might' counterfactual illustrated with the dashed area in Subfigure~\ref{subfig:mightmin}. However, since for the `Existential Might' counterfactual it is enough to find just one closest antecedent-world satisfying $\psi$, we can actually remove all worlds except $W$, because there is a $\psi$-world directly adjacent to it.

\section{Temporal Counterfactuals}\label{sec:decidability}

We now leverage the semantic insights garnered on counterfactuals in the previous section to design a logic for expressing notions of counterfactual causality in the temporal domain. For this, it is not enough to solely evaluate the truth values of counterfactual conditionals or even their minimal versions, since causality commonly places additional conditions for a causal relationship to hold. For instance, cause and effect may have to be satisfied in the reference world~\cite{HalpernP01,Halpern15,CoenenFFHMS22}. The syntax and semantics of our logic are introduced in the following section, where we also establish decidability of satisfiability and trace checking. In Section~\ref{sec:ex}, we close by illustrating how our logic can express definitions of causality proposed in previous literature.

\subsection{QPTL With Counterfactuals}

We consider a logic that builds Boolean combinations of \qptl formulas and the proposed counterfactual conditionals relating \qptl statements. Nesting of counterfactuals is not allowed, although may be interesting to explore in future work to study temporally structured interventions. We call our logic \qptlc, for \emph{\qptl with counterfactuals}, and its formulas are built according to the following grammar:
\begin{align*}
\xi \Coloneqq ~&\psi \cf \psi \mid \psi \mcf \psi \mid \psi \ucf \psi \mid \psi \emcf \psi \mid \psi \cfmin \psi \mid\\ & \psi \mcfmin \psi \mid \psi \ucfmin \psi \mid \psi \emcfmin \psi \mid
\xi \land \xi \mid \lnot \xi \mid \psi \enspace ,
\end{align*}
where $\psi$ is a \qptl formula. The semantics of the counterfactual conditionals are as discussed in the previous section where worlds are now infinite traces and the satisfaction relation is as for \qptl, and the semantics of the \qptl-formulas are as described in Section~\ref{sec:prelim}.

Now that we have fixed the syntax and semantics of our logic \qptlc, we show that important decision problems can be automatically decided. Note that previous results on Lewis' original two counterfactuals and related works on causal analysis either consider theories with a finite set of models~\cite{Lewis71}, endorse the problematic limit assumption \cite{CoenenFFHMS22}, or both~\cite{Halpern15,AleksandrowiczC17}. In this sense, our results extend previous work in several orthogonal directions.

We start by considering the satisfiability problem of \qptlc. Satisfiability denotes the problem of deciding whether there exists a trace that satisfies a given formula. 
Deciding the satisfiability of counterfactuals of course depends on whether the similarity relation can be expressed in a decidable logic. We, therefore, show satisfiability of minimal counterfactuals with respect to the following family of \qptl-expressible similarity relations:
\begin{align*}
        \leq^{\texttt{QPTL}} = \{\, R \subseteq \big((2^\mathit{AP})^\omega\big)^3 \mid \, &\exists \text{ \qptl-formula } \varphi_R \text{ over } \mathit{AP}_{\{\pi_1,\pi_2,\pi_3\}} \ldot 
        \forall t_1,t_2,t_3 \in (2^\mathit{AP})^\omega \ldot \\& R(t_1,t_2,t_3) \text{ iff } t_1 \cupdot t_2 \cupdot t_3 \models \varphi_R
        \} \enspace,
\end{align*}

where $\mathit{AP}_P = \{a_\pi \mid a \in \mathit{AP} \land \pi \in P\}$ indexes the atomic propositions with trace variables and $ t_1 \cupdot t_2 \cupdot t_3$ denotes the disjoint fusion of the three traces: We have for all positions $n \in \mathbb{N}$ that the following holds: $(t_1 \cupdot t_2 \cupdot t_3)[n] = \{a_{\pi_1} \mid a \in t_1[n] \} \cup \{a_{\pi_2} \mid a \in t_2[n] \} \cup \{a_{\pi_3} \mid a \in t_3[n] \}$. Note that the similarity relations $\leq^{\qptl}_t\!\!(X)$ introduced in Example~\ref{ex:distance} are subsumed by this family, and also the distance metric considered by Coenen et al.~\cite{CoenenFFHMS22}. We consider three-place relations here because the first place models the $W$-subscript of the relation $\leq_W$.

Satisfiability of formulas in \qptlc can then be decided with an idea roughly similar to the reduction of \ltl satisfiability to \ltl model checking proposed by Rozier and Vardi~\cite{RozierV10}. However, since there is no corresponding trace-based logic that can express the semantics of counterfactual conditionals, we instead encode the problem as a \hyperqptl model-checking problem over a model containing all possible traces over our set of atomic propositions $\mathit{AP}$.  In the end, since \hyperqptl model checking is decidable~\cite{Rabe16}, so is the satisfiability of \qptlc-formulas. This approach is interesting from a theoretical point of view, because \hyperqptl satisfiability checking itself is undecidable. We leverage the fact that models of \qptlc formulas are traces and not sets of traces. While the logic does have relational semantics, these are effectively guarded through the counterfactual conditionals, allowing us to encode into model checking and avoid the undecidable \hyperqptl satisfiability problem.

\begin{theorem}\label{thm:sat}
For any \qptlc formula $\varphi$, it is decidable to check  whether $\varphi$ is satisfiable when the similarity relation is from the family $\leq^{\texttt{QPTL}}$ and the universe is defined by a \qptl formula.
\end{theorem}

Since our logic contains negation, this also covers the problem of checking validity of a formula in the usual way, as a formula is valid if and only if its negation is not satisfiable. Similarly, because our logic subsumes \qptl, we can also model-check \qptlc-properties on any system that can be expressed in \qptl, which covers many practically relevant finite-state systems. This can be done by checking satisfiability of the conjunction of the system formula and the negation of the property formula (and possibly choosing the system as the universe).

Lastly, we consider the problem of checking whether some \qptlc-formula holds on a given trace. Since we need a finite representation of the infinite trace to feasibly compute this truth value, we consider lasso-shaped traces, i.e., traces of the form $t = t_0 \ldots t_i(t_{i+1} \ldots t_j)^\omega$ that ultimately repeat some loop in the infinite $\omega$-part. The proof of the theorem then follows a similar idea as the satisfiability proof, except that we do not search for any satisfying trace and instead fix the corresponding trace variable to the lasso-shaped trace.

\begin{theorem}\label{thm:trace_check}
For a lasso-shaped trace $t$ and a \qptlc formula $\varphi$, it is decidable to check  whether $t \models \varphi $ when the similarity relation is from the family $\leq^{\texttt{QPTL}}$ and the universe is defined by a \qptl formula.
\end{theorem}

Theorem~\ref{thm:sat} and Theorem~\ref{thm:trace_check} show that it is possible to build tools that automatically check whether causal relationships hold in a given finite-state system, or whether causal relationships expressed in our logic are present on a given trace. While an exact complexity analysis is out of scope of this paper, \qptl has non-elementary complexity that scales in the number of propositional quantifier alternations. The complexity of our decision procedure for \qptlc as of now additionally scales in the number of counterfactual conditionals, since these introduce trace quantifier alternations and the complexity of \hyperqptl model checking scales with this number. We believe there may be a more efficient encoding, as there is no dependence between the trace variables resulting from different counterfactual conditionals and, hence, quantifiers can be rearranged to avoid alternations. However, the purpose of this work is to explore how expressive our temporal logic with counterfactuals can be while retaining decidability. In practice, it may in fact be more feasible to consider fragments of \qptl that have a more practical complexity, such as \ltl. There exist efficient model checkers~\cite{BeutnerF23a}  for its counterpart \texttt{HyperLTL}~\cite{ClarksonFKMRS14} and, since recently, also for \texttt{HyperQPTL}~\cite{BeutnerF23b}, such that implementing a model checker for \qptlc or a fragment modulo \texttt{LTL} using our outlined encoding is feasible.

\subsection{Examples}\label{sec:ex}

We close by illustrating how \qptlc can be used to express several notions of counterfactual causality proposed in previous literature. We start with Lewis' account of counterfactual causation to illustrate how a basic, but limited, notion of causation can be expressed solely with counterfactuals. We then outline at the example of Halpern's (non-temporal) actual causality the important role of proper modeling in causal analysis, and lastly we show how previous work of Coenen et al.\ on extending Halpern's actual causality to temporal properties is subsumed by our work on temporal-counterfactual reasoning with \qptlc.

\paragraph{Lewis' Counterfactual Causation~\cite{Lewis73b}.}

Based on counterfactual conditionals, Lewis introduced a definition of causality in order to capture Hume's informal account that ``we may define a cause to be an object followed by another, where, if the first object had not been, the second never had existed''~\cite{Hume1748}. Lewis only draws causal relationships between two events, which are rather informally defined objects, but which in previous literature on traces are commonly interpreted to mean the value of an atomic proposition at a certain position~\cite{BeerBCOT12,Leitner-FischerL13,CoenenDFFHHMS22}. Lewis definition stipulates that an event $c$ is the cause for an event $e$ if the following condition holds:
$(c \cf e) \land  (\lnot c \cf \lnot e)$.
The intuition behind the formula is that either the cause $c$ and $e$ appear together (and hence the left conjunct is also satisfied), and then moving to any closest $\lnot c$-world is sufficient to avoid the effect $e$ (corresponding to the left conjunct), or, neither cause and effect appear, and then moving to a $c$-world is sufficient to bring about the effect (the mirrored case). Consequently, this causal relation holds even when cause and effect are not present in a world, and hence may be model checked on a system to infer whether it holds globally.
 
With \qptlc, we can improve on Lewis' original formulation in several ways. First, we can use the `Universal Would' counterfactual and hence need not assume a total similarity relation.  Further, Lewis' original logic lacks facilities for temporal reasoning and, hence, cannot express that a cause is ``an object followed by another''. In \qptlc, it is possible to express that the cause has to preceed the effect, and the remaining temporal requirements, in the following way:

$$\big((\LTLeventually c \land \LTLeventually e) \rightarrow (\lnot e \LTLuntil c)\big) \land (\LTLeventually c \ucf \LTLeventually e) \land  (\lnot \LTLeventually c \ucf \lnot \LTLeventually e) \enspace .$$

We can deduce causal relationships by checking this formula for validity in a universe of interest. It ensures that if the cause and effect appear on some trace, then, particularly, the cause happens before the effect. While this solves the lacking temporal expressivity in Lewis' logic, it still shares the idiosyncratic features of this definition of causality which have been raised in the literature since. For instance, if a cause has two effects, the earlier effect is also considered a cause for the later effect.

These peculiarities cannot be solved through counterfactual reasoning alone and over the years several solutions in the causal modeling literature have emerged, for instance \emph{interventions} and \emph{contingencies}. In the following paragraph, we show that our framework can emulate these concepts by modeling their mechanics in the universe and the similarity relation. 

\paragraph{Halpern's Actual Causality~\cite{Halpern15}.}

	We outline a direct correspondence of our approach to Halpern's modified version~\cite{Halpern15} of actual causation~\cite{HalpernP01}, for binary variables. The exact definition of this notion of causality is given in Appendix~\ref{app:actual} due to space reasons. The general idea is that the dynamics of the universe are defined by a structural equation $f_x$ for every variable $x$, which may depend on other variables and a set of external variables which are defined by a context. We assume these to be Boolean functions in the case of binary variables. Halpern restricts the analysis to acyclic dependencies. Based on some model $\mathcal{M} = (\mathcal{S},\mathcal{F})$ and context $c$ we may assume an evaluation function $U^\mathcal{M}_c(x)$ which tells us for any variable $x \in V$ whether it evaluates to $1$ or $0$. To establish a correspondence to \qptlc, we construct a universe for a specific context under analysis based on basic equivalences for every variable: $x \leftrightarrow ((f_x \land \lnot i_x) \lor c_x)$ if  $U^\mathcal{M}_c(x) = 1$, else $x \leftrightarrow ((f_x \lor i_x) \land \lnot c_x)$. If some $f_x$ depends on external variables, we substitute them by the appropriate Boolean constant depending on the context. Then, $\mathcal{U}(\mathcal{M},c)$ denotes the conjunction of all of these equivalences together with $\LTLnext \LTLglobally (\bigwedge_{v \in V} \lnot v \land \lnot i_v \land \lnot c_v)$, i.e., all traces have only empty sets after the first position. This suffices because actual causation has no particular temporal aspect, so we may effectively model all outcomes of the structural equations on a finite trace prefix of length one. Let $t^\mathcal{M}_c \in \mathcal{L}(\mathcal{U}(\mathcal{M},c))$ be the unique trace where for all $n > 0: t^\mathcal{M}_c[n] = \emptyset$ and for all $x \in V: \{i_x,c_x\} \cap  t^\mathcal{M}_c[0] = \emptyset$, i.e., no interventions or contingencies are present. It may seem problematic that we have to construct the universe based on given observations, but note that even in structural equations, certain equations have to be manipulated for modeling interventions and contingencies.
	
	\begin{theorem}\label{thm:hp_to_cf}
		$X_0 = x^0_0 \land \ldots \land x^k_0 ,\ldots, X_n = x^0_n \land \ldots \land x^j_n$ are actual causes of $\varphi$ in $(\mathcal{M},\vec{u})$, iff $t^\mathcal{M}_c \models  \varphi_X \land \varphi \land \lnot \varphi_X \mcfmin \lnot \varphi$ in the universe $\mathcal{U}(\mathcal{M},c)$ with respect to the similarity relation $\leq_t\!\!(\{i_x \mid x \in V\})$ that tracks only the active interventions between traces, where 
		$$ \lnot \varphi_X = (\lnot x^0_0 \land \ldots \land \lnot x^k_0) \lor \ldots \lor (\lnot x^0_n \land \ldots \land \lnot x^j_n)  \text{ is in Blake canonical form~\cite{Blake1937}.}$$
	\end{theorem}

The intuition behind the proof is that the closest traces to $t^\mathcal{M}_c$ that satisfy $\lnot \varphi_X$ require negating one of the disjuncts, and hence require flipping the values of all of the variables appearing there. This simulates interventions on these variables. Existential search for contingencies is taken care of by the `Might' counterfactual modality, since all possible contingencies for some intervention form an equivalence class under $\leq_t\!(\{i_x \mid x \in V\})$.
The minimality of the counterfactual ensures that $\lnot \varphi_X$ cannot be enlarged. This means neither can a conjunct be added to $\varphi_X$, and hence indeed all causes are described, nor can a disjunct be removed, and hence all causes are minimal. We give an illustrative example of this encoding in the following.

	\begin{example}
		Consider the classic example of determining whether lightning (l) or a dropped match (m) caused a forest fire (f). Assume that only the match was dropped and that this suffices to cause the fire because the equation for the fire is $f := l \lor m$ and the other two variables are determined by outside factors. We construct our universe based on the concrete observations $f = 1$, $l = 0$, and $m=1$ and model interventions and contingencies accordingly:
		\begin{align*}
			\mathcal{U}_{\mathit{fire}}(\mathcal{M},c) = \, &(l \leftrightarrow (i_l \land \lnot c_l)) \land 
			(m \leftrightarrow (\lnot i_m \lor c_m)) \\
			&\land \, \big(f \leftrightarrow (((l \lor m) \land \lnot i_f) \lor c_f)\big) \enspace .
		\end{align*} 
		Note that the actual causes according to Halpern's definition in this scenario are $f = 1$, i.e., the effect itself, and $m = 1$. We have that the formula $f \land m$ qualifies as a cause for the effect $f$ on the trace $t_\mathit{fire} = \{f,m\}\emptyset^\omega$ describing the above scenario. This is because we have $t_\mathit{fire} \models f \land m \land (\lnot f \lor \lnot m) \mcfmin \lnot f$ in $\mathcal{U}_{\mathit{fire}}(\mathcal{M},c)$.
	\end{example}
	
	All conjunctions appearing in the Blake canonical form in Theorem~\ref{thm:hp_to_cf} are prime implicants. Dubslaff et al.~\cite{DubslaffWBA22} have used these before to compute feature causes, which are counterfactual in nature, in configurable software systems. We use prime implicants here to establish an intriguing formal connection between property-based causes and event-based causes in structural equation models. This result is valuable because event-based causes have a restrictive and explicit syntax. The formula of Theorem~\ref{thm:hp_to_cf} allows characterizing the same counterfactual reasoning with more expressive languages for causes and provides a formal basis for property-based extensions of actual causality. We will discuss such an extension by Coenen et al.\ in the following paragraph.

\paragraph{Coenen et al.'s Temporal Causality~\cite{CoenenFFHMS22}.}

Coenen et al.\ extend interventions and contingencies to reactive systems described by Moore automata, and in this way lift Halpern's actual causality to temporal properties on traces. The key idea is that, if the reference trace $t$ is given in a lasso-shape, then the behavior of contingencies can be modeled in a finite-state machine called counterfactual automaton $\mathcal{C}_t^\mathcal{T}$, where $\mathcal{T}$ is the original Moore automaton. This corresponds to the construction for finite settings in the above paragraph. Interventions need not be modeled explicitly because Coenen et al.'s causality only characterizes causes on the input sequence of the traces and as there are no dependencies between inputs, just changing them outright suffices to model interventions. Therefore, the concrete distance metric used in that work is a modification of $\leq_t\!\!(\mathcal{I})$, where $\mathcal{I}$ is the set of inputs. The modification is done in order to satisfy the Limit Assumption, but results in a coarse over-approximation of the set of closest traces. For instance, for $t = \{a\}^\omega$, the closest traces satisfying $\LTLeventually \LTLglobally \lnot a$ with Coenen et al.'s similarity relation are the whole set $\mathcal{L}(\LTLeventually \LTLglobally \lnot a)$, while our work does not require the Limit Assumption to hold and can use the unmodified $\leq_t\!\!(\mathcal{I})$, which has an infinite chain of ever closer traces: $\{\}^\omega,\{a\}\{\}^\omega,\{a\}\{a\}\{\}^\omega,$ and so on.

It turns out that when we characterize Coenen et al.'s definition of causality in \qptlc, we see a minor divergence from Halpern's actual causality. We can encode Coenen et al.'s causality in this way: $\varphi$ is a temporal cause for $\psi$ on $t$, iff $$t \models \varphi \land \psi \land \big((\lnot \varphi \ucfmin \lnot \psi) \lor (\lnot \varphi \mcfmin \lnot \psi)\big) \enspace .$$
in the universe defined by the traces of $\mathcal{C}_t^\mathcal{T}$. We include the `Universal Would' counterfactual because we need to make use of its vacuous satisfaction mechanics, i.e., $\top$ may qualify as a temporal cause according to Coenen et al.'s definition by virtue of quantification over an empty set, but $\bot$ will never work as an antecedent in $\mcfmin$, since there needs to be a trace satisfying it. We can use the compositional nature of \qptlc to emulate this by using $\ucfmin$, because it is, except for the vacuous case, a stronger condition, i.e., whenever $\lnot \varphi \ucfmin \lnot \psi$ holds on a trace $t$ in a universe and a $\lnot \varphi$-trace exists, also $\lnot \varphi \mcfmin \lnot \psi$ holds on trace $t$. We believe that the compositionality of a logic like \qptlc can be a useful tool for comparing different definitions of causality, as demonstrated here between Halpern's actual causality and Coenen et al.'s temporal causality. This process can even  be automated with the outlined \qptlc decision procedures.

\section{Conclusion}
In this paper, we study a fusion of two prominent flavors of modal logic: counterfactual and temporal reasoning. Our theoretical results are a step towards the automatic evaluation of temporal counterfactual conditionals on infinite sequences, such as counterexample traces returned by a model checker or trajectories of a reinforcement learning agent. Further, our extension of Lewis' theory of counterfactual conditionals to non-total similarity relations and our minimal counterfactual operators are relevant to the theory of counterfactuals beyond the temporal reasoning considered in this work. In the future, we plan on using our logic to define system-level and trace-level causation in reactive systems, which correspond to the notions of global and actual causation. An interesting question here is, whether system-level semantics of counterfactuals should be a universal application of the trace semantics, or should counterfactually relate different system models. The latter approach may utilize previous work on system mutations~\cite{KupfermanLS08} studied in the area of coverage~\cite{ChocklerKV01}, which has a tight relationship to counterfactual causality~\cite{ChocklerHK08}. We are also interested in automating the discovery of causal relationships between temporal properties on infinite sequences. This problem can be framed as synthesizing satisfying antecedents of counterfactuals.

\bibliographystyle{abbrv}
\bibliography{bibliography}

\appendix

\section{Additional Preliminaries}

\subsection{HyperQPTL}\label{app:hyperqptl}

The syntax of \hyperqptl~\cite{Rabe16,CoenenFHH19} is given as follows, where $\psi$ is a \qptl formula. For sake of simplicity, we only present a fragment of \hyperqptl without alternation of propositional and trace quantification. This fragment suffices for the results presented in this work.
\begin{equation*}
	\chi \Coloneqq \exists \pi \ldot \chi \mid \forall \pi \ldot \chi \mid \psi \enspace .
\end{equation*}
The semantics of \hyperqptl are not defined over single traces, but sets of traces. The alphabet of atomic propositions $\mathit{AP}_{\texttt{Hyper}} = \{a_\pi \mid a \in \mathit{AP} \land \pi \in \mathcal{V} \}$, therefore, is indexed by trace variables. The satisfaction relation is defined with respect to a time point~$i$, a set of traces~$\Tr$ and a trace assignment $\traceassign: \tracevars \rightarrow \Tr$ that maps trace variables to traces. 
To update the trace assignment so that it maps trace variable~$\pi$ to trace~$t$, we write $\traceassign[\pi \mapsto t]$. We further lift our replacement function $t[q \mapsto t_q]$ to sets of traces such that $T[q \mapsto t_q] = \{t[q \mapsto t_q] \mid t \in T \}$.
\begin{equation*}
	\begin{array}{lll}
		\traceassign,i \models_\Tr a_\pi       & \text{iff } & a \in \traceassign(\pi)[i] \\
		\traceassign,i \models_\Tr q         & \text{iff } & \forall t \in \Tr \ldot q \in t[i] \\
		\traceassign,i \models_\Tr \neg \varphi              & \text{iff } & \traceassign,i \nmodels_\Tr \varphi \\
		\traceassign,i \models_\Tr \varphi \land \psi         & \text{iff } & \traceassign,i \models_\Tr \varphi \text{ and } \traceassign,i \models_\Tr \psi \\
		\traceassign,i \models_\Tr \X \varphi                & \text{iff } & \traceassign,i+1 \models_\Tr \varphi \\
		\traceassign,i \models_\Tr \varphi\U\psi             & \text{iff } & \exists j \geq i \ldot \traceassign, j \models_\Tr \psi \; \land \forall i \leq k < j \ldot \traceassign,k \models_\Tr \varphi \\
		
		\traceassign,i \models_\Tr \forall \pi \ldot \varphi & \text{iff } & \forall t \in \Tr \ldot \traceassign[\pi \mapsto t],i \models_\Tr \varphi \\
		\traceassign,i \models_\Tr \exists \pi \ldot \varphi & \text{iff } & \exists t \in \Tr \ldot \traceassign[\pi \mapsto t],i \models_\Tr \varphi\\
		\traceassign,i \models_\Tr \forall q \ldot \varphi & \text{iff } & \forall t_q \in (2^{\{q\}})^\omega \ldot \traceassign,i \models_{\Tr[q \mapsto t_q] } \varphi \\
		\traceassign,i \models_\Tr \exists q \ldot \varphi & \text{iff } & \exists t_q \in (2^{\{q\}})^\omega \ldot \traceassign,i \models_{\Tr[q \mapsto t_q] } \varphi \\\enspace 
	\end{array}
\end{equation*}
A set of traces $T$ satisfies a \hyperqptl formula $\varphi$ iff $\emptyset,0 \models_T \varphi $, which we also denote by $T \models \varphi$.

\begin{example}
	With \hyperqptl, we can relate traces to one another. For instance, the following formula characterizes sets where all pairs of traces (and hence all traces) have the same action sequence:
	$$\forall \pi \ldot \forall \pi' \ldot \LTLglobally (u_\pi \leftrightarrow u_{\pi'}) \land \LTLglobally (d_\pi \leftrightarrow d_{\pi'}) \enspace .$$
	It is satisfied by the singleton set $\{t\}$ with trace $t$ from Example~1, but not by $\{t,t'\}$ where $t' = (\{b,d\})^\omega$, since, e.g., the actions at the first position of the traces differ.
\end{example}

\subsection{Actual Causality}\label{app:actual}
Actual causality was originally proposed by Halpern and Pearl \cite{HalpernP01}. Several improvements have been appeared since, we consider the latest of these proposed by Halpern \cite{Halpern15}. A \emph{causal model} $\mathcal{M} = (\mathcal{S},\mathcal{F})$ is defined by a \emph{signature} $\mathcal{S}$ and set of \emph{structural equations} $\mathcal{F}$. A~signature $\mathcal{S}$ is a tuple $(E,V,R)$, where $E$ is a set of \emph{exogenous} variables, $V$ is a set of \emph{endogenous} variables, and $R$ defines  the \emph{range} of possible values $R(Y)$ for all variables $Y \in E\cup V$. For some context $\vec{u}$, the value of an exogenous variable is determined by factors outside of the model, while the value of some endogenous variable $X$ is defined by the associated structural equation $f_X \in \mathcal{F}$.

\begin{definition}\label{def:hpcausality}
	$\vec{X} = \vec{x}$ is an \emph{actual cause} of $\varphi$ in $(\mathcal{M},\vec{u})$, if the following 
	holds.
	\begin{description}
		\item AC1: $(\mathcal{M},\vec{u}) \models \vec{X} = \vec{x}$ and $(\mathcal{M},\vec{u}) \models \varphi$, i.e., cause and effect are true in the actual world, and
		\item AC2: There is a set $\vec{W}$ of variables in $V$ and a setting $\vec{x}'$ of the variables in $\vec{X}$ such that if $(\mathcal{M},\vec{u}) \models \vec{W} = \vec{w}$, then $(\mathcal{M},\vec{u}) \models [\vec{X} \leftarrow \vec{x}',\vec{W} \leftarrow \vec{w} ] \lnot \varphi$, and
		\item AC3: $\vec{X}$ is minimal, i.e. no subset of $\vec{X}$ satisfies AC1 and AC2.
	\end{description}
\end{definition}

In the case of binary variables, we may denote a cause $\vec{X} = \vec{x}$ with $X = x^0 \land \ldots \land x^n$, where for $0 \leq i \leq n$ a literal $x^i$ is positive if it evaluates to 1 in $ \vec{x}$, and negative if not. We denote by $V(x^i)$ the variable $v \in V$ associated with literal $x_i$.

\section{Proofs}

\setcounter{counter}{8}

\begin{proposition}
	If $\leq_W$ is a total order, then `Universal Would' and `Would' are equivalent, i.e., for all pairs of properties $\varphi$ and $\psi$ we have that $\varphi \ucf \psi \equiv \varphi \cf \psi$.
\end{proposition}

\begin{proof}
	We prove the equivalence by proving the entailment in each direction separately. So first assume that there exists a world $W$ that satisfies $\varphi \ucf \psi$. We show that this world also satisfies $\varphi \cf \psi$. Consider two cases: First let us assume $\varphi \ucf \psi$ is vacuously satisfied by $W$ such that there exists no world $W_1$ with $W_1 \models \varphi$. Then $W$ also vacuously satisfies $\varphi \cf \psi$. As the second case, consider the non-vacuous case such that there exists at least one world $W'$ that satisfies $\varphi$. Then, $\varphi \ucf \psi$ is satisfied because, for all worlds $W_1$ that satisfy $\varphi$, there exists an at least equally close world $W_2$ that also satisfies $\varphi$ such that all closer worlds $W_3$ satisfy $\varphi \rightarrow \psi$. Hence, such a $W_2$ world exists in particular for $W'$, and can serve as a witness for the existential quantifier in the semantics of $\cf$ (Condition~2 of Definition~5). Hence, $W \models \varphi \cf \psi$.
	
	For the entailment in the other direction, assume that there is a world $W$ satisfying $\varphi \cf \psi$. We show that this world also satisfies $\varphi \ucf \psi$. Again, if $W$ satisfies $\varphi \cf \psi$ vacuously it also satisfies $\varphi \ucf \psi$ vacuously. We, therefore, assume that there is a world $W_1$ such that $W_1 \models \varphi$ and all at least equally close worlds $W_2$ satisfy $\varphi \rightarrow \psi$. We now show that this is enough to satisfy Condition~2 of Definition~9. Pick any world $W'$ as an instantiation of the outermost universal quantifier in that condition. If $W' \nmodels \varphi$ it trivially satisfies the implication in the quantifiers body, so assume $W' \models \varphi$. Since $\leq_W$ is a total order, we know that either $W' \leq_W W_1$ or $W_1 \leq_W W'$. If $W' \leq_W W_1$, we know that all smaller worlds $W'' \leq_W W'$ satisfy $\varphi \rightarrow \psi$ due to transitivity of $\leq_W$. Hence, $W'$ is a witness for the existential quantifier in Condition~2 of Definition~9. Lastly, assume $W_1 \leq_W W'$. Then $W_1$ is the witness for the existential quantifier. In any case, we have $W \models \varphi \ucf \psi$ which closes this direction.
	
\end{proof}

\addtocounter{counter}{1}

\begin{corollary}
	If $\leq_W$ is a total order, then `Existential Might' and `Might' are equivalent, i.e., for all pairs of properties $\varphi$ and $\psi$ we have that $\varphi \emcf \psi \equiv \varphi \mcf \psi$.
\end{corollary}

\begin{proof}
	The result follows directly from the two duality laws $\varphi \cf \psi \equiv \lnot (\varphi \mcf \lnot \psi)$ and $\varphi \ucf \psi \equiv \lnot (\varphi \emcf \lnot \psi)$, and Proposition 10, through substitution.
\end{proof}

\addtocounter{counter}{2}

\begin{lemma}
	For a minimal counterfactual statement $\varphi \rightsquigarrow_{\!_\mathit{min}} \! \psi$, with $\rightsquigarrow \; \in \{ \cf,\mcf,\ucf,\emcf \}$, there exists a parameterized formula $\varphi_{\mathit{FO}}(W)$, with $W \in \mathcal{U}$, that quantifies only over worlds and not properties such that for all worlds $W' \in \mathcal{U}$:
	$W' \models \varphi \rightsquigarrow_{\!_\mathit{min}} \psi$ iff $\varphi_{\mathit{FO}}(W')$ is valid.
\end{lemma}

\begin{proof}
	
	The high-level idea of the proof is that, iff $\varphi$ is not the minimal antecedent and there, therefore, exists a formula $\varphi'$ characterizing a superset of the worlds characterized by $\varphi$, then there exists a set of worlds $S$ that can be added to $\varphi$ without changing the truth value of the counterfactual conditional.
	Such worlds fall into two categories: they are further away from the reference world $W$ than the closest $\varphi$-worlds and can be added without considering whether they satisfy $\psi$, or they are at least equally close to $W$ than the closest $\varphi$-world but also satisfy $\psi$ so they could be added without harm and would constitute new closest $\varphi$-worlds. The exact relationships these worlds have to satisfy with respect to the closest $\varphi$-worlds differ between the operators, further, a difference is whether the reasoning has to be extended to incomparable worlds on other chains in the similarity relation. Therefore, we split the argumentation for the four minimal operators from now. We mark the world-variables that correspond to the closest $\varphi$-worlds of interest with $W^c$ in each, and else use variables $W'$ for $W_1$ to establish a close correspondence to definitions of the semantics. All other world variables are specific to the minimality reasoning. If not stated otherwise, quantifiers quantify over worlds from the universe $\mathcal{U}$. We have that $W$ satisfies $\varphi \cfmin \psi$ iff 
	\begin{align}
		\big(&(\forall W' \ldot W' \nmodels \psi) \land (\forall W^h \in \Bar{\mathcal{U}} \ldot \forall W' \in \mathcal{U} \ldot W^h \models \varphi \land W' \nmodels \varphi)\big) \lor \Big(\exists  W^c \ldot W^c \models \varphi \land \forall W'' \ldot\label{cfmin:eq1}\\ 
		&\big(W'' \leq_W W^c \Rightarrow W'' \models \varphi \leftrightarrow \psi \land (W'' \models \psi \Rightarrow \forall W^i \ldot W^i \not\leq_W W'' \Rightarrow W^i \models \varphi) \big) \, \land\label{cfmin:eq2} \\
		&\big( W'' \not\leq_W W^c \land W'' \models \psi \land (\exists W^n \ldot W^n \leq_W W^c \land W^n \not\leq_W W'' \land W^n \models \lnot \varphi) \Rightarrow\label{cfmin:eq3}\\ &\exists W^p \ldot W^p \leq_W W'' \land W^p \models \varphi \big)
		\Big)\label{cfmin:eq4}
	\end{align}
	is valid. The first disjunct in Line~\ref{cfmin:eq1} encodes the fact that iff $\psi$ is unsatisfiable with respect to the universe, then $\varphi$ should characterize exactly the complement of the universe.  The right disjunct encodes that we can enlarge the set characterized by $\varphi$ based on the following conditions in relation to the counterfactual world $W^c$ (corresponds to $W_1$ in Condition~2 of Definition~5):
	\begin{itemize}
		\item There is an at least equally close or closer world $W''$ such that $\psi$ holds but this world is not included in $\varphi$ yet (strengthens the semantic by $W'' \models \psi \rightarrow \varphi$ as in Line~\ref{cfmin:eq2}), or in $W''$ $\psi$ holds but further away or incomparable worlds are not included in $\varphi$ but could be (Rest of Line~\ref{cfmin:eq2}).
		\item There is a world at least as close as $W^c$ which could take its place in a more minimal antecedent, i.e., $W''$ satisfies $\psi$, all its at-least-as-close worlds are a proper subset of the ones of $W^c$ (Line~\ref{cfmin:eq3}) and are not included in $\varphi$ (Line~\ref{cfmin:eq4}). 
	\end{itemize}
	Next, we have that $W$ satisfies $\varphi \mcfmin \psi$ iff 
	\begin{align}
		(&\exists W'. W' \models \varphi) \land \big(\forall W' \ldot W' \models \varphi \Rightarrow \exists W^c . W^c \leq_W W' \land (W^c \models \varphi \land \psi) \, \land\label{mcfmin:eq1}\\ &\forall W''' \ldot (W^c \not\leq_W W''' \Rightarrow W''' \models \lnot\varphi \rightarrow \lnot\psi) \land (W^c \leq_W W''' \Rightarrow W''' \models \varphi)\big)\label{mcfmin:eq2}
	\end{align}
	is valid. Here, we have a fairly direct strengthening of the semantics of regular `Might' (Line~\ref{mcfmin:eq1}): We now additionally require for $\varphi$-world $W'$ on an infinite chain not only an at least equally closest world $W^c$ (this time corresponds to $W_2$ in Condition~2 of Definition~7)  that satisfies $\varphi \land \psi$, we also require that $\varphi$ includes all smaller $\psi$-worlds and hence place $\lnot \varphi \rightarrow \psi$ as a requirement on them. If this did not hold, we could include them in the property. Additionally, we could include any world on a different chain (and which is hence not comparable to $W^c$) if they are not yet included in $\varphi$ but satisfy $\psi$. Hence, such worlds also have to satisfy $\lnot \varphi \rightarrow \psi$. Both is covered by the left conjunct in Line~\ref{mcfmin:eq2}. Additionally, we again require that worlds further away than $W^c$ are included in the property (right disjunct in Line~\ref{mcfmin:eq2}). If they were not we could include them while still retaining $W^c$ as a witness for all of them.
	
	\medskip
	
	\noindent Next, we have that $W$ satisfies $\varphi \ucfmin \psi$ iff 
	\begin{align}
		\big(&(\forall W' \ldot W' \nmodels \psi) \land (\forall W^h \in \Bar{\mathcal{U}} \ldot \forall W' \in \mathcal{U} \ldot W^h \models \varphi \land W' \nmodels \varphi)\big) \, \lor\label{uw:eq1}\\ \big(&\exists W^h \ldot W^h \models \varphi \land \forall W' \ldot W' \models \varphi \Rightarrow \exists W^c \ldot W^c \leq_W W' \land W^c \models \varphi \land \forall W''' \ldot\label{uw:eq2}\\ &(W''' \leq_W W^c \Rightarrow W''' \models \varphi \leftrightarrow \psi) \land (W^c \leq_W W''' \Rightarrow W''' \models \varphi)\label{uw:eq3}\\ &\land (W''' \not\leq_W W^c \land W^c \not\leq_W W''' \Rightarrow W''' \models \lnot \varphi \rightarrow\lnot \psi )\big)\label{uw:eq4}
	\end{align} is valid. This is now because of a combination from the reasoning of the previous two operators. Like for $\cfmin$, we again have that if $\psi$ is unsatisfiable $\varphi$ has to characterize the complement of the universe (Line~\ref{uw:eq1}). If this is not the case, then there has to be at least one $W^h$ world satisfying $\varphi$ (Line~\ref{uw:eq2}). We have to ensure this because $\varphi$ may only characterize an effectively empty set if there exists no $\psi$-world. Further, for any $\varphi$-world there exists a $\varphi$-world $W^c$ closest to $W$ such that on any closer worlds $\varphi \rightarrow \psi$. Again, like for $\cfmin$, we add the inverse direction $\psi \rightarrow \varphi$ to ensure that $\varphi$ cannot be enlarged in this direction (left in Line~\ref{uw:eq3}). Further, any world further away than $W^c$ has to be included in $\varphi$ as they can all use $W^c$ as a closest $\varphi$-world (right in Line~\ref{uw:eq3}). Lastly, like for $\mcfmin$, if some world incomparable to $W^c$ which hence is on a different chain satisfies $\psi$, it also has to be included in $\varphi$ which we ensure by requiring $\lnot \varphi \rightarrow \lnot \psi$ (Line~\ref{uw:eq4}).
	
	\medskip
	
	\noindent Lastly, we have that $W$ satisfies $\varphi \emcfmin \psi$ iff 
	\begin{align}
		&\exists W' \ldot W' \models \varphi \land \forall W'' \ldot W'' \leq_W W'  \land W'' \models \varphi \Rightarrow\exists W^c \ldot W^c \leq_W W'' \land (W^c \models \varphi \land \psi) \, \land\label{em:eq1}\\
		&\forall W^h \ldot (W^h \leq_W W^c \Rightarrow W^h \models \lnot \varphi \rightarrow \lnot \psi )\land  \big( W^h \not\leq_W W^c \Rightarrow W^h \models \varphi \land W^h \models \psi \land\label{em:eq2} \\
		&(\exists W^n \ldot W^n \leq_W W^c \land W^n \not\leq_W W^h \land W^n \models \lnot \varphi) \Rightarrow\exists W^p \ldot W^p \leq_W W^h \land W^p \models \varphi \big)\label{em:eq3}
	\end{align} 
is valid. In Line~\ref{em:eq1}, we have the usual semantics of `Existential Might'. We now additionally require that worlds at least as close as the closest $\varphi$-worlds $W^c$ that satisfy $\psi$, which correspond to $W_3$ in Definition~\ref{def:emight_cf}, are included in $\varphi$ if they satisfy $\psi$ (left conjunct in Line~\ref{em:eq2}). Further, we require all worlds farther away or incomparable to $W^c$ to be included in $\varphi$ (right in conjunct in Line~\ref{em:eq2}). Similar to the reasoning for $\cfmin$, we have to ensure that these worlds $W^h$ do not induce a more minimal antecedent by qualifying as a $W^c$ by satisfying $\psi$ (end of Line~\ref{em:eq2}) and having a proper subset of at-least-equally-close worlds (Line~\ref{em:eq3}).
\end{proof}

\begin{theorem}
For any \qptlc formula $\varphi$, it is decidable to check  whether $\varphi$ is satisfiable when the similarity relation is from the family $\leq^{\texttt{QPTL}}$ and the universe is defined by a \qptl formula.
\end{theorem}

\begin{proof} We sketch how to encode the satisfiability of our \qptlc-formulas in a \hyperqptl model-checking problem over the most general model $\mathcal{M}$ that contains all traces over our alphabet $2^\mathit{AP}$. Let $\varphi^\mathcal{U}$ be the formula encoding the universe. 
	The semantics of the non-minimal operators contain quantification over all possible traces over our alphabet $2^\mathit{AP}$ and hence over the traces from the model $\mathcal{M}$. In Lemma~14 we showed that the semantics of the minimal operators can similarly be expressed by quantification over traces from $\mathcal{M}$ and in $\varphi^\mathcal{U}$ as well. We now illustrate at the example of $\mcfmin$ how such a formula can be encoded into a \hyperqptl-formula $\varphi^\qptl_\mathit{FO}(\pi)$ in prenex normal form, as even if $\mathit{AP} = \emptyset$, there exists at least one trace. Let $\varphi_R$ bet the \qptl-formula for the similarity metric $R \in \leq^\qptl$. For some formula $\varphi$ and trace variable $\pi \in \mathcal{V}$, $\varphi_\pi$ denotes the same formula where all atomic propositions are indexed with trace variable $\pi$. Then $\varphi^\qptl_\mathit{FO}(\pi)$ for some $\varphi \mcfmin \psi$ is of the following form:
	\begin{align*}
		\varphi^\qptl_\mathit{FO}(\pi) = &\exists \pi_1 \ldot \forall \pi_2 \ldot \exists \pi_3 \ldot \forall \pi_4 \ldot  \varphi^\mathcal{U}_{\pi_1} \land \varphi_{\pi_1} \land \bigg(\varphi^\mathcal{U}_{\pi_2} \rightarrow \Big(\varphi_{\pi_2} \rightarrow \varphi^\mathcal{U}_{\pi_3} \land\varphi_R(\pi,\pi_3,\pi_2) \, \land \varphi_{\pi_3} \land \psi_{\pi_3} \land\\ 
		&   \big(\varphi^\mathcal{U}_{\pi_4} \rightarrow (\lnot\varphi_R(\pi,\pi_3,\pi_4) \rightarrow (\lnot \varphi_{\pi_4} \rightarrow \lnot \psi_{\pi_4})) \land (\varphi_R(\pi,\pi_3,\pi_4) \rightarrow \varphi_{\pi_4})\big)\Big)\bigg) \enspace .
	\end{align*}
	We restrict quantification to the universe with conjunctions and implications containing $\varphi^\mathcal{U}$. Quantification over the complement of the universe uses $\lnot \varphi^\mathcal{U}$ instead.
	
	For the original \qptlc-formula $\varphi$, we apply this transformation to any counterfactual conditional appearing in it and again transform the result into prenex normal form. This is possible while retaining equisatisfiability because temporal operators appear only on the lowest level, i.e., no quantifier appears in the scope of a temporal operator. We denote the resulting formula with $\varphi^\mathit{full}_\mathit{FO}(\pi)$ and put a $\pi$-subscript on any top-level \qptl formula (i.e., that was not in the body of a counterfactual operator). 
	
	We can then model check the formula \hyperqptl-formula $\exists \pi \ldot \varphi^\mathit{full}_\mathit{FO}(\pi)$ on the model $\mathcal{M}$ to check whether there exists a trace satisfying our original \qptlc-formula. Since \hyperqptl model checking is decidable, so is the satisfiability checking of \qptlc. This approach has similarities to the idea of reducing the satisfiability problem of linear-time temporal logic (\texttt{LTL}) to model checking of the same logic~\cite{RozierV10}. Note that, crucially, our approach is possible because the models of our logic \qptlc are traces and not sets of traces as in \hyperqptl, which allows an encoding into the decidable model checking problem while avoiding the undecidable problem of \hyperqptl satisfiability checking.
\end{proof}

\begin{theorem}
	For a lasso-shaped trace $t$ and a \qptlc formula $\varphi$, it is decidable to check  whether $t \models \varphi $ when the similarity relation is from the family $\leq^{\texttt{QPTL}}$ and the universe is defined by a \qptl formula.
\end{theorem}

\begin{proof}(Sketch)
	Since the trace $t$ is lasso-shaped, we encode it in a formula $\varphi^t$ that characterizes the set $\{t\}$. We can then use the ideas of the satisfiability proof above to solve the trace checking problem. However, we now model check the formula $\forall \pi \ldot \varphi^t_\pi \rightarrow \varphi^\mathit{full}_\mathit{FO}(\pi)$ on the model $\mathcal{M}$, not searching for any trace satisfying the formula but instead fixing the reference trace to $t$.
\end{proof}

	\begin{theorem}
	$X_0 = x^0_0 \land \ldots \land x^k_0 ,\ldots, X_n = x^0_n \land \ldots \land x^j_n$ are actual causes of $\varphi$ in $(\mathcal{M},\vec{u})$, iff $t^\mathcal{M}_c \models  \varphi_X \land \varphi \land \lnot \varphi_X \mcfmin \lnot \varphi$ in the universe $\mathcal{U}(\mathcal{M},c)$ with respect to the similarity relation $\leq_t\!\!(\{i_x \mid x \in V\})$ that tracks only the active interventions between traces, where 
	$$ \lnot \varphi_X = (\lnot x^0_0 \land \ldots \land \lnot x^k_0) \lor \ldots \lor (\lnot x^0_n \land \ldots \land \lnot x^j_n)  \text{ is in Blake canonical form~\cite{Blake1937}.}$$
\end{theorem}

\begin{proof} This is effectively a finite setting modeled on the first position of a trace, such that the Limit Assumption holds and for any formula $\lnot \varphi_X$ without temporal operators, there exists a unique (finite) set of closest traces from $t^\mathcal{M}_c$ that satisfy it. Let $I = \{ i_x \mid x \in V\}$ denote the set of intervention variables. We will use the following auxiliary results throughout the proof:
	\begin{enumerate}
		\item\label{aux1} For any actual cause $X_0 = x^0_0 \land \ldots \land x^k_0$ for effect $\varphi$, there is a contingency $W$ such that  $(\mathcal{M},\vec{u}) \models [X\leftarrow \lnot x^0_i \land \ldots \land \lnot x^k_i,W \leftarrow w^0_p \land \ldots \land w^q_p ] \lnot \varphi$, i.e., the intervention that flips all the variables is a witnessing intervention for AC2. With Boolean variables, this follows directly from minimality, as any variable in an intervention that is not negated could be moved to the contingency $W$ instead, yielding a more minimal cause.
		
		\item\label{aux2} There is a direct mapping from interventions that flip all values to intervention variables on traces: By induction over the number of structural equations we can show that for all traces $t \in \mathcal{U}(\mathcal{M},c)$ and variables $v \in V$, we have $(\mathcal{M},\vec{u}) \models [X\leftarrow x^0_i \land \ldots \land x^k_i,W \leftarrow w^0_p \land \ldots \land w^q_p ] v$ where for all $x^r_i \in X: (\mathcal{M},\vec{u})\models \lnot x^r_i $, iff $v \in t[0]$, where $t$ is the unique trace that has all the contingency variables of  $W$ enabled at the first position, i.e., for all $0 \leq r \leq q: V(w^r_p) \in t[0]$, and all intervention variables corresponding to literals in $X$ are also enabled: for all $0 \leq r \leq k: (x^k_p \leftrightarrow V(x^k_p) \in t^\mathcal{M}_c) \rightarrow i_{V(x^k_p)} \in t[0]$.
		\item\label{aux3} Let $t_\mathit{cf} $ be a closest $\lnot \varphi_X$-trace to $t^\mathcal{M}_c$ such that $t_\mathit{cf}  \models (\lnot x^0_z \land \ldots \land \lnot x^j_z)$ for some conjunction and $t_\mathit{cf}  \models \lnot \varphi$. We can use induction over the length of the conjunction to show that for all $i \in I: i \in t_\mathit{cf} [0]$ iff there is some  $0 \leq r \leq z: i_{V(x^r_z)}= i \land  c_{V(x^r_z)} \not\in t_\mathit{cf} [0]$, i.e., exactly the intervention variables corresponding to literals in the conjunction are enabled on $t_\mathit{cf} $, while all corresponding contingency variables are disabled. 
	\end{enumerate}
	We now proceed to proof the two directions of the equivalence separately.\\
``$\Rightarrow$'': We know that no intervention or contingency variables are true on $t^\mathcal{M}_c$ by construction. With Result~\ref{aux2} and the fact that $X_0,\ldots, X_n$ are actual causes of $\varphi$ in $(\mathcal{M},\vec{u})$, it follows from AC1 that $t^\mathcal{M}_c \models \varphi_X \land \varphi$. To show that $\lnot \varphi_X \mcfmin \lnot \varphi$ is satisfied, consider that from Result~\ref{aux2} and Result~\ref{aux3}, it follows the closest traces from  satisfying $\lnot \varphi_X$ have to differ in all of the values in at least one of the conjuncts of $\varphi_X$, w.l.o.g.\ assume this to be $(x^0_0 \lor \ldots \lor x^n_0)$. With  Result~\ref{aux1} we have $(\mathcal{M},\vec{u}) \models [X_0\leftarrow \lnot x^0_i \land \ldots \land \lnot x^k_i,W_0 \leftarrow w^0_p \land \ldots \land w^q_p ] \lnot \varphi$, and with Result~\ref{aux2} it then follows that there is a trace $t$ such that $t \models \lnot \varphi$, and from Result~\ref{aux3} it follows that $t$ is a closest $\lnot \varphi_X$-trace. Hence, We have $t^\mathcal{M}_c \models \lnot \varphi_X \mcf \varphi$. To show that $\lnot \varphi_X$ is also a minimal antecedent, consider what would happen if a trace was added to $\mathcal{L}(\lnot \varphi_X)$. Such a trace $t''$ has to satisfy $\varphi_X$, and hence satisfies one of the literals in each of that formula's conjuncts. That means that the changes between $t^\mathcal{M}_c$ and $t''$ are either incomparable to any other closest $\lnot \varphi_X$-trace, or a proper subset. In the former case, this induces another actual cause, in the latter case, this would mean one of the causes is not minimal. In both cases, we have a contradiction. Note that $\lnot \varphi_X$ is in Blake canonical form because all of the causes are minimal, and hence prime implicants, i.e., no  $X_v$ implies another $X_w$, and this also holds for their negations in $\lnot \varphi_X$. This closes this direction\\
	``$\Leftarrow$'': We first show AC1 holds for all of the causes. We know that $(\mathcal{M},\vec{u}) \models \varphi$ from $t^\mathcal{M}_c \models \varphi$ and Result~\ref{aux2}.  To show that $t^\mathcal{M}_c \models (x^0_i \land \ldots \land x^u_i)$ for all $X_i$, consider $t_\mathit{cf}$ as defined in Result~\ref{aux3}. Since we know from the proof of Result~\ref{aux3} that on $t_\mathit{cf}$ exactly the intervention variables corresponding to the literals in the conjunction are set, and intervention variables flip the value of variables with respect to $t^\mathcal{M}_c$, we can deduce that all the literals occurring in the conjunction have the inverted value in $t^\mathcal{M}_c$, which proves AC1 for all of the causes using Result~\ref{aux2}. AC2 follows directly from $t^\mathcal{M}_c \models \lnot \varphi_X \mcfmin \varphi$, Result~\ref{aux3} which states on some closest trace $t_\mathit{cf} \models \lnot \varphi$ all literals $x^0_z \land \ldots \land x^j_z$ are intervened upon, and Result~\ref{aux2} which relates the intervention and contingency variables  $w^0_z \land \ldots \land w^q_z$ of $t_\mathit{cf}$ to the causal model $(\mathcal{M},\vec{u})$. Then, we can deduce that $(\mathcal{M},\vec{u}) \models [X \leftarrow \lnot x^0_z \land \ldots \land \lnot x^j_z, W \leftarrow w^0_z \land \ldots \land w^q_z]\lnot \varphi$. Lastly, we can show AC3 by contradiction: If one of the causes was not minimal, then Result~\ref{aux2} and Result~\ref{aux3} imply there is closer $\lnot \varphi$-trace, and therefore $\lnot \varphi_X$ is not a minimal antecedent, a contradiction. 
\end{proof}


\end{document}